\documentclass[final,leqno,onefignum,onetabnum]{siamltex1213}

\usepackage[usenames,dvipsnames]{xcolor}
\usepackage{framed}
\usepackage{comment}

\includecomment{Long}
\includecomment{Extended}
\excludecomment{Short}
\usepackage[toc,page]{appendix}
\usepackage{graphicx}
\usepackage{cite}
\usepackage{url}
\usepackage{amsmath}
\usepackage{amssymb}
\usepackage{pgf}
\usepackage{pgfplots}
\usepackage{tikz}
\usetikzlibrary{arrows,automata}
\definecolor{light-gray}{gray}{0.8}
\newtheorem{cor}{Corollary}
\newtheorem{example}{Example}
\newtheorem{conj}[theorem]{Conjecture}
\newcommand{\argmax}{\operatornamewithlimits{argmax}}
\newcommand\xqed[1]{%
  \leavevmode\unskip\penalty9999 \hbox{}\nobreak\hfill
  \quad\hbox{#1}}
\newcommand\demo{\xqed{$\triangle$}}

\usepackage{microtype}

\usepackage{url}
\providecommand{\com[2]}{\begin{tt}[#1: #2]\end{tt}}

\providecommand{\rmj[1]}{{\color{black}#1}}
\providecommand{\fg[1]}{{\color{black}#1}} 
\providecommand{\rev[1]}{{\color{black}#1}}

\begin{document}

\markboth{F. Gonze \and R.M. Jungers}
{SPF and TRT approaches for synchronizing automata}

\title{
On the Synchronizing Probability Function and the Triple Rendezvous Time for Synchronizing Automata
}

\author{Fran\c cois Gonze\thanks{
This work was also supported by the communaut\'e francaise de Belgique - Actions de Recherche Concert\'ees and by the Belgian Program on Interuniversity Attraction Poles initiated by the Belgian Federal Science Policy Office.} 
\and Rapha\"el M. Jungers \thanks{R. M. Jungers is a F.R.S.-FNRS Research Associate}}

\maketitle

%\titlerunning{SPF and $T_3$ approaches to \v Cern{\'y}'s Conjecture} 

%\setcounter{footnote}{0} 
\begin{abstract}
\v Cern{\'y}'s conjecture is a longstanding open problem in automata theory. We study two different concepts, which allow to approach it from a new angle. The first one is the \emph{triple rendezvous time}, i.e., the length of the shortest word mapping three states onto a single one. The second one is the \emph{synchronizing probability function} of an automaton, a recently introduced tool which reinterprets the synchronizing phenomenon as a two-player game, and allows to obtain optimal strategies through a Linear Program.

Our contribution is twofold. First, by coupling two different novel approaches based on the synchronizing probability function and properties of linear programming, we obtain a new upper bound on the triple rendezvous time. 
Second, by exhibiting a family of counterexamples, we disprove a conjecture on the growth of the synchronizing probability function. We then suggest natural follow-ups towards \v Cern{\'y}'s conjecture.
\end{abstract}



\section{Synchronizing Automata and \v Cern{\'y}'s Conjecture}

\makeatletter{\renewcommand*{\@makefnmark}{}
\footnotetext{Preliminary results have been presented at the conference LATA 2015 \cite{GonzeLATA}.}\makeatother}

Automata are a natural way to model systems that can take multiple different states, so that actions made on these systems have an effect depending on the current state. For such systems, it can be desirable to have a particular input sequence which would ensure a known final state, independently of the initial one. Automata with this property are called \emph{synchronizing}. Synchronizing automata first appeared in computers and relay control systems in the 60s. In the 80s and 90s, this subject found applications in robotics and in the industry. More recently, it has been used to model consensus theory and linked with primitivity of matrix sets (see \cite{PYChev, JungersBlondelOlshevsky14}).

Formally, a \emph{deterministic, finite state, complete automaton} (DFA) is a triplet $(Q,\Sigma, \delta)$ with $Q$ the set of \emph{states}, $\Sigma$ the alphabet of \emph{letters} and $\delta$ the transition function $\delta: Q\times \Sigma \rightarrow Q$ defining the effect of the letters on the states. Figure \ref{ExampleAutomaton} shows an example of such an automaton. In this paper, we will represent automata as sets of matrices as follows. 
\begin{figure}[h!]
\begin{center}
\scalebox{1}{
\begin{tikzpicture}[->,>=stealth',shorten >=1pt,auto,node distance=1.7cm,
                    semithick]
  \tikzstyle{every state}=[fill=light-gray,draw=none,text=black, scale=0.8]

  \node[state] (A)                    {$1$};
  \node[state]         (B) [above right of=A] {$2$};
  \node[state]         (D) [below right of=A] {$4$};
  \node[state]         (C) [below right of=B] {$3$};

  \path (A) edge              node {b} (B)
            edge [loop left]  node {a} (C)
        (B) edge [loop above] node {a} (B)
            edge              node {b} (C)
        (C) edge              node {b} (D)
         	edge [loop right] node {a} (D)
        (D) edge  			  node {a, b} (A);
\end{tikzpicture}
}
\end{center}
\caption{A synchronizing automaton. The word $abbbabbba$ maps every state onto state $1$.}
\label{ExampleAutomaton}
\end{figure}
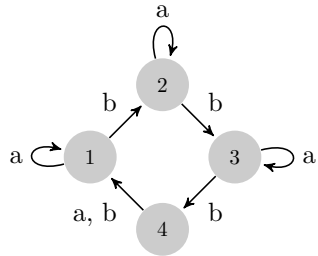
%\vspace{0.5 cm}
\begin{definition}
A \textit{(deterministic, finite state, complete) \emph{automaton}} (DFA) is a set of $m$ row-stochastic matrices $\Sigma \subset \{0,1\}^{n\times n}$ (where $m,n$ are respectively the number of \emph{letters} in the alphabet, and the number of \emph{states} of the automaton). Each letter corresponds to a matrix $L\in \Sigma$ with binary entries, which satisfies  $ L e ^T=e^T,$ where $e$ is the $1\times n$ all-ones vector. We write $\Sigma^t$ for the set of products of length $t$ of matrices taken in $\Sigma.$  We refer to these matrices as \emph{words} of length $t$. States and sets of states are represented by their $1\times n $ characteristic vector in the canonical way.
\end{definition}
%\vspace{0.5 cm}

\begin{definition}
An automaton $\Sigma \subset \{0,1\}^{n\times n}$ is \emph{synchronizing} if there is an index $1\leq i\leq n$ and a finite product $W=L_{c_1}\fg{\cdots} L_{c_s}:\, L_{c_j}\in \Sigma$ which satisfy $$W=e^Te_i, $$
where $e_i$ is the $i$th standard $1\times n$ basis vector.\\
In this case, the sequence of letters $L_{c_1}\fg{\cdots} L_{c_s}$ is said to be a \emph{synchronizing word}.
\end{definition}
\begin{example}
\label{ExampleLetters}
The two letters of the automaton in Fig.\ref{ExampleAutomaton} are the following matrices:
\rev{\arraycolsep=1.8pt\def\arraystretch{0.8}
$$a=\left( \begin{array}{cccc}
1 & 0 & 0 & 0\\[-0.05cm]
0 & 1 & 0 & 0 \\[-0.05cm]
0 & 0 & 1 & 0 \\[-0.05cm]
1 & 0 & 0 & 0 \end{array} \right)
b=\left( \begin{array}{cccc}
0 & 1 & 0 & 0\\[-0.05cm]
0 & 0 & 1 & 0 \\[-0.05cm]
0 & 0 & 0 & 1 \\[-0.05cm]
1 & 0 & 0 & 0 \end{array} \right)$$}
and we have the synchronizing word \arraycolsep=1.8pt\def\arraystretch{0.8} $abbbabbba=\left(\begin{array}{cccc}
1 & 1 & 1 & 1 \end{array}\right)^T\left(\begin{array}{cccc}
1 & 0 & 0 & 0 \end{array}\right).$\demo
\end{example}

Verifying that an automaton is synchronizing can be done in quadratic time\footnote{When we refer to the length of a word or to computational complexity, we compare it to the number of states $n$ of the automaton.}. However, finding a short synchronizing word is hard (see \cite{OlUm}). Jan \v Cern{\'y} conjectured in 1964 \cite{cernyPirickaRosenauerova64} that if an automaton is synchronizing, the length of its shortest synchronizing word, \rev{also called \emph{reset threshold} of the automaton}, is quadratic: 
%\vspace{0.5 cm}
\begin{conj}[\v Cern{\'y}'s conjecture, 1964 \cite{cernyPirickaRosenauerova64}] \label{cernyconj}
Let $\Sigma \subset \{0,1\}^{n\times n}$ be a synchronizing automaton.  Then, there is a synchronizing word of length at most $(n-1)^2.$\end{conj}
%\vspace{0.5 cm}
Although this conjecture is simple to state, it is still unsolved. If the conjecture is true, then $(n-1)^2$ is also a tight bound. Indeed, in  \cite{cerny64}, \v Cern\'y proposes an infinite family of automata attaining it, for any number of states. We refer to this family as the \emph{\v Cern\'y family} of automata. Figure \ref{ExampleAutomaton} shows the automaton of the family with four states.
Synchronizing automata attaining this bound or having a shortest synchronizing word close to it are very infrequent \rev{(see \cite{AnanGusVolk, JKari01, Roman, Ananichev200730} for examples)}. On a brighter side, another longstanding open problem based on synchronizing automata, namely, the road-coloring problem, was recently solved by Trahtman (see \cite{trahtman}). Many problems mixing road-coloring and \v Cern\'y's conjecture are still open (see \cite{Vorel, ThesisSp}). 

%\vspace{0.5 cm}

In the last decades, Conjecture \ref{cernyconj} has been the subject of intense research. 
It has been proven to hold for several families of automata (see  \cite{kari03,eppstein90reset, Dubuc98, trahtman_cerny,1065155,BBP11,carpiarticle,cerny64}), including cyclic and Eulerian.
However, the best general upper bound \rev{on the reset threshold} of an automaton with $n$ states is equal to $ (n^3-n)/6$, obtained by Pin and Frankl \cite{Frankl82, Pin83a}, and rediscovered independently in \cite{KlyachkoRystsovSpivak87}. 
This bound has been holding for more than 30 years\footnote{A bound of $n(7n^2+6n-16)/48$ was proposed in \cite{Trahtman2011}, but its proof was incorrect, as shown in \cite{GJT2014}.}. A state of the art overview is given by Volkov \cite{volkov_survey}.

%\vspace{0.5 cm}

Recently, several research efforts have tried to shed light on the problem by making use of probabilistic approaches (see \cite{jungers_sync_12, steinberg-2009}). The main tool we will focus on, the \emph{Synchronizing Probability Function} (SPF), was introduced by the second author in 2012 \cite{jungers_sync_12}. This tool allows the reformulation of the synchronizing property as a game theoretical problem whose solution can be obtained through convex optimization. Convex optimization is a mature discipline with strong theoretical basis (see for instance \cite{BoydBook}). Our hope is that in this framework, properties of synchronization can be better understood and proved using tools that have not been used on DFA yet.

The philosophy behind Conjecture \ref{cernyconj} is to bound the length of the shortest word mapping all the states onto a single one. Based on this idea, one could wonder what is the length of the shortest word for which there exists a set of states \emph{of a given cardinality} mapped onto a single state by this word. For cardinality two the problem is solved, as in any synchronizing automaton there always exists a single letter mapping two states onto a single state. Therefore in that case the answer is \rev{``one''}. For higher values, to the best of our knowledge, the question is open. We will analyse the case of cardinality three, which we coin the \emph{triple rendezvous time}.

%\vspace{0.5 cm}

 In Section 2 we recall the main properties of the synchronizing probability function. In Section 3, we introduce the concept of triple rendezvous time and, by making use of the synchronizing probability function, we obtain a new upper bound on this value. In Section 4, we refute a recent conjecture on the synchronizing probability function (Conjecture 2 in \cite{jungers_sync_12}) by presenting a particular family of automata which does not satisfy it. This paper is the journal version of our conference presentation appeared in \cite{GonzeLATA}, with examples, full proofs and an improved upper bound for the triple rendezvous time (Theorem \ref{HighT3bound}).

\section{A Game Theoretical Framework and the Synchronizing Probability Function}

In this section, we recall definitions and properties of the synchronizing probability function needed to develop our results. A more complete introduction to the SPF and the details of the proofs can be found in \cite{jungers_sync_12}. This concept is based on the following two-player game, which gives another perspective on the synchronization of an automaton. For a given automaton and a length $t$ chosen in advance, the rules are as follows:
\begin{enumerate}
\item Player Two secretly chooses a state $e_j$ of the automaton.
\item Player One chooses a word $W$ of length at most $t$. 
\item Player One guesses the final state $e_jW$. If it is the right state, he wins. Otherwise, Player Two wins.
\end{enumerate}

In this game, if the length $t$ is larger or equal to the length of a synchronizing word, a winning strategy for Player One is to take this synchronizing word and the state on which the automaton is mapped. Oppositely, if $t$ is zero, Player One can only choose randomly one of the states, and he \rev{has} probability $1/n$ of winning. We consider that both players can choose probabilistic strategies.
%\vspace{0.5 cm}

 The \emph{policy of Player Two} is defined as a probability distribution over the states, that is, any vector $p\in \mathbb{R}^{n }_+, ep^T=1.$  Player Two chooses the state \rev{$e_j$} with probability \rev{$p_j,$} in which case the automaton will end up at the state corresponding to \rev{$e_jW$}. Since Player One wants to maximize the probability of choosing the right final state, he will pick up the state where the probability for the automaton to end is maximal, that is,
  $$\rev{\argmax_i (pW)e_i^T.} $$ 
Therefore the probability of winning for Player One is
  \begin{equation}
  \label{eqn-proba}
  \rev{\max_{i,W}(pW)e_i^T.}
 \end{equation}
The aim of Player Two is to minimize that probability.
  
In the following, $\Sigma^{\leq t}$ is the set of products of length at most $t$ of matrices taken in $\Sigma.$  By convention, and for the ease of notation, it contains the product of length zero, which is the identity matrix.

\begin{definition}[SPF, Definition 2 in \cite{jungers_sync_12}]
Let $n\in \mathbb N$ and $\Sigma \subset \{0,1\}^{n\times n}$ be an automaton.
 The \emph{ synchronizing probability function} (SPF) of $\Sigma$ is the function $k_\Sigma:\, \mathbb N \rightarrow \mathbb R_+:$
\rev{\begin{eqnarray}\label{eqn-catchingproba}
k_\Sigma(t)&=&\min_{p\in \mathbb R^{ n}_+,\ ep^T=1}{\left \{\max_{W\in \Sigma^{\leq t}, i}(pW)e_i^T\right \}}.\end{eqnarray}}\end{definition}
The SPF gives the probability of winning the game that Player One can achieve with parameter $t$ if Player Two plays optimally. If there is no ambiguity on the automaton, we use $k(t)$ for $k_\Sigma(t)$. 

The synchronizing probability function is non-decreasing with respect to $t$. Moreover, its value is one if and only if there is a synchronizing word of length smaller or equal to $t$. This leads to the following reformulation of Conjecture \ref{cernyconj}:

\begin{proposition}[Proposition 1 in \cite{jungers_sync_12}]
The following conjecture is equivalent to Conjecture \ref{cernyconj}:

If  $\Sigma \subset \{0,1\}^{n\times n}$ is a synchronizing automaton, then,
$$\forall t\geq (n-1)^2,\quad k_\Sigma(t)=1 .$$\end{proposition}
Figure \ref{ExampleSPF} represents the SPF of the automaton presented in Fig. \ref{ExampleAutomaton}.
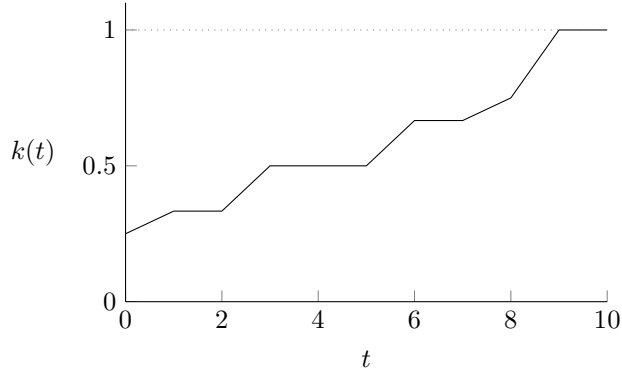
\begin{figure}[h!]
\begin{center}
% This file was created by matlab2tikz v0.4.7 running on MATLAB 7.14.
% Copyright (c) 2008--2014, Nico Schlömer <nico.schloemer@gmail.com>
% All rights reserved.
% Minimal pgfplots version: 1.3
% 
% The latest updates can be retrieved from
%   http://www.mathworks.com/matlabcentral/fileexchange/22022-matlab2tikz
% where you can also make suggestions and rate matlab2tikz.
% 
\begin{tikzpicture}

\begin{axis}[%
width=2.52083333333333in,
height=1.565625in,
scale only axis,
xmin=0,
xmax=10,
xlabel={$t$},
ymin=0,
ymax=1.1,
ylabel={$k(t)$},
ylabel style={rotate=-90},
axis x line*=bottom,
axis y line*=left
]
\addplot [color=black,solid,forget plot]
  table[row sep=crcr]{%
1	0.333333333313348\\
2	0.333333332953387\\
3	0.499999999999815\\
4	0.499999999999773\\
5	0.500000001117456\\
6	0.666666666662096\\
7	0.666666666488453\\
8	0.750000005937906\\
9	1.00000000000043\\
10	1.00000000000964\\
};
\addplot [color=gray,dotted,forget plot]
  table[row sep=crcr]{%
0	1\\
10	1\\
};
\addplot [color=black,solid,forget plot]
  table[row sep=crcr]{%
0	0.25\\
1	0.333333333333333\\
};
\end{axis}
\end{tikzpicture}

\end{center}
\caption{The synchronizing probability function of the automaton with four states presented in Fig. \ref{ExampleAutomaton}. There is a synchronizing word of length nine, therefore $k(9)=1$.}
\label{ExampleSPF}
\end{figure}
In order to use the SPF, we need an explicit algorithmic construction of the optimal strategies for both players, which allows us to compute its value. Each basic strategy of Player One, i.e., the choice of a word and a final state, is equivalent to choosing a column in this word. Therefore, we consider the set of all the different columns reached in words of length at most $t$.

\begin{definition}
\label{DefA}
We say that a vector is \emph{reachable at $t$} if it is equal to any column of the words in $\Sigma^{\leq t}$. We denote by $A(t)$ the set of all the reachable vectors at $t$. We represent $A(t)$ as a $n\times m(t)$ matrix, where $m(t)$ is the number of different reachable vectors at $t$.
\end{definition}

We notice that $A(t-1)\subseteq A(t)$. In order to have a unique matrix representation for $A(t)$, $t\geq 1$, we choose the first block of $A(t)$ equal to $A(t-1)$, we sort the $m(t)-m(t-1)$ last columns of $A(t)$ by lexicographical order, and we choose $A(0)$ equal to the identity matrix. When there is no ambiguity on $t$, we use $A$ for $A(t)$.

\begin{example}
\label{A3}
Let us go back to the automaton presented in Example \ref{ExampleLetters}. At $t=3$, $A(3)$ is given by: $$A(3)=\left( \begin{array}{ccccccc}
1 & 0 & 0 & 0 &  1 & 0 & 0\\[-0.05cm]
0 & 1 & 0 & 0 &  0 & 0 & 1\\[-0.05cm]
0 & 0 & 1 & 0 &  0 & 1 & 1\\[-0.05cm]
0 & 0 & 0 & 1 &  1 & 1 & 0\end{array} \right).$$
The first four columns corresponds to $A(0)$, the fifth column comes from the word $a$, the sixth from the word $ba$ and the seventh from the word $bba$.
\demo
\end{example}

The \emph{policy of Player One} is defined as a probability distribution over the columns of $A(t)$, that is, any column vector $q\in \mathbb{R}^{m(t)}_+$ such that $eq=1.$

It turns out that the SPF can be efficiently computed thanks to the following linear programs\footnote{The following inequalities are entrywise.}. 

%\vspace{0.5 cm}
 
 \begin{theorem}[Theorem 1 in \cite{jungers_sync_12}]
  The \emph{synchronizing probability function} $k_\Sigma(t)$ of a DFA $\Sigma$ is given by
\begin{eqnarray}
\label{eqn-catchingprobalin}
\min_{p,k}& &\ k\\ \nonumber s.t. \ 
&&{{pA\leq ke}}\\\nonumber && ep^T=1\\\nonumber &&p\geq 0 .
\end{eqnarray}
 It is also given by:
 \begin{eqnarray}
 \label{eqn-catchingprobalindual}
 \max_{q,k}&&\ k\\
\nonumber  s.t. \ &&{{Aq\geq ke^T}}  \\\nonumber &&eq=1\\\nonumber &&q \geq 0 .
\end{eqnarray} 
\rev{In the equations above, $A$ denotes the set of reachable vectors at $t$ (see Definition \ref{DefA}), $p$ is a $1 \times n$ vector}, $q$ is a $m(t)\times 1$ vector, $e$ represents all-ones vectors of the appropriate dimension, 1 is a scalar, and 0 represents zero vectors of the appropriate dimensions.\end{theorem}

The linear Program (\ref{eqn-catchingprobalindual}) is the dual of Program (\ref{eqn-catchingprobalin}). In the primal (\ref{eqn-catchingprobalin}), the optimal objective value $k(p)$ is obtained with the strategy of Player Two, $p$, which is a probability distribution on the states. In the dual (\ref{eqn-catchingprobalindual}), the optimal objective value $k(q)$ is obtained with the strategy of Player One, $q$, which is a probability distribution on the set of reachable vectors. For any primal feasible solution $p$ and any dual feasible solution $q$, the objective value $k(p)$ of Program (\ref{eqn-catchingprobalin}) and the objective value $k(q)$ of Program (\ref{eqn-catchingprobalindual}) satisfy $k(p)\geq k(q)$. Therefore, if the objective value $k$ is the same for both programs with feasible solutions $p$ and $q$, this value is the optimum (see \cite{BoydBook} for more details on convex optimization and linear programming).

\begin{example}
\label{A3Solution}
Let us consider the automaton of Fig. \ref{ExampleAutomaton} and $t=3$. The set of reachable vectors is given in Example \ref{A3}. On the one hand $p=(1/4, 1/4, 1/4, 1/4)$ is an admissible solution for Program (\ref{eqn-catchingprobalin}) (i.e., a probability distribution on the \rev{states} for Player Two), which gives as objective value $k(p)=1/2$. On the other hand $q=(0, 0, 0, 0, 1/2, 0, 1/2)^T$ is an admissible solution for Program (\ref{eqn-catchingprobalindual}) (i.e., a probability distribution on the \rev{columns of words in $\Sigma^{\leq 3}$}), which also gives the objective value $k(q)=1/2$. Therefore, the SPF at $t=3$ is equal to $k(3)=1/2$ (as shown in Fig. \ref{ExampleSPF}). In other words, this means that if both players play optimally, with words of length at most 3, Player One has \rev{probability $1/2$ of winning the game}.
\demo
\end{example}
%\vspace{0.5 cm} 

By leveraging classical results from convex optimization, one can derive strong properties on the optimal strategies in the above-defined game.

 \begin{theorem}[Theorem 2 in \cite{jungers_sync_12}]
\label{THM2}
For any pair of optimal solutions $(p^*(t), q^*(t))$ of Programs (\ref{eqn-catchingprobalin}) and (\ref{eqn-catchingprobalindual}), we have

$q^*_j(k-(p^{*}A)_j)=0$ for all $1\leq j \leq m(t)$ and

$p^*_i((Aq^*)_i-k)=0$ for all $1\leq i \leq n.$
 
 \end{theorem}

In the following, our main arguments will be based on the dimension of the set of optimal strategies of Program (\ref{eqn-catchingprobalindual}):

 \begin{definition}[Definition 3 in \cite{jungers_sync_12}]
 Let $\Sigma$ be an automaton and $t$ be a positive integer. The polytopes $P_t$ and $Q_t$ are the sets of optimal solutions of respectively Program (\ref{eqn-catchingprobalin}) and Program (\ref{eqn-catchingprobalindual}). We call \emph{dimension} of a polytope the dimension of the smallest affine subspace containing the polytope.
 \end{definition}

\begin{example}
In Example \ref{A3}, $P_3$ and $Q_3$ are the following sets:
\rev{\begin{equation}
\nonumber
P_3=
\left\lbrace\begin{array}{l|l}
&-1/4\leq x \leq 1/4,\\[-0.05cm]
(1/4+x, 1/4-x, 1/4+y, 1/4-y)& -1/4\leq y \leq 1/4,\\[-0.05cm]
& x-y\leq 0
\end{array}\right\rbrace 
\end{equation}
\begin{eqnarray}
\nonumber
Q_3&=&\{(0, 0, 0, 0, 1/2, 0, 1/2)^T\}.
\end{eqnarray}}
These are the only solutions allowing for an objective value $k=1/2$.
\demo
\end{example}

Moreover, if the SPF does not increase, we have the following result on $P_t$:

\begin{lemma}[Lemma 1 in \cite{jungers_sync_12}]
\label{inclusion}
If $k(t)=k(t+1)$ then $P_{t+1}\subset P_t$.\end{lemma}

With these tools in hand, we can get to our contributions.

\section{A New Bound on the Triple Rendezvous Time}
The \emph{\rev{triple rendezvous time (TRT)}} is the length of the shortest word mapping three states of the automaton onto a single one. Although it is a very natural concept, we are not aware of any attempts to bound its value for synchronizing automata. In what follows, the \emph{weight} of a vector is the number of its non-zero elements.

 \begin{definition}
 For a synchronizing automaton $\Sigma$, the \emph{triple rendezvous time} $T_{3,\Sigma}$ is defined as the smallest integer $t$ such that $A(t)$ contains a column of weight superior or equal to 3.
 \end{definition}

In other words, \rev{the} triple rendezvous time is the length of the shortest word $W$ such that there exist states $q_i$, $q_j$ and $q_k$ with $q_iW=q_jW=q_kW$. In the following, we will use $T_3$ for $T_{3,\Sigma}$ when there is no ambiguity on the automaton.

Of course, we can extend this concept to $T_l$, the length of the shortest word which maps $l$ states onto a single one, i.e., the length of the shortest word $W$ such that there exist $l$ states $q_{i}, q_{j}, \dots $ such that $q_{i}W=q_{j}W=\dots$. We notice that $T_n$ is the length of the shortest synchronizing word, and that for any synchronizing automaton, $T_2=1$.

\rev{Our motivations for studying the triple rendezvous time are numerous. First, it is a natural problem related to Conjecture \ref{cernyconj}, and it allows for new approaches to study synchronization properties of an automaton. Second, the TRT is directly linked with the evolution of the synchronizing probability function (Proposition 6 and Conjecture 4 in \cite{jungers_sync_12}), and is also related to the $k-$extension property developed in \cite{Berlinkov-extension}\footnote{In the terminology of \cite{Berlinkov-extension}, \rev{the TRT} can be defined as the smallest integer such that there is a pair of states which are synchronized by some single letter, and which is \rev{$(T_3-1)-extendable$}.
}. Third, the triple rendezvous time can also be used as an indicator to find automata with large reset threshold. Indeed, for many classes of automata, the value of the TRT seems empirically to be correlated with the length of the shortest reset word: for random synchronizing automata (from the framework of \cite{skvortsov2011experimental}), the reset threshold is small and $T_3=1$  with high probability\footnote{In the framework of \cite{skvortsov2011experimental}, automata are synchronizing with high probability and have with high probability three states mapped on a single one by one letter, which implies that $T_3=1$ with high probability.}. Oppositely, all the automata known in the literature which achieve the bound of Conjecture \ref{cernyconj} do have a large \rev{TRT}\footnote{The reader can easily check that the automata of the \v Cern{\'y} family have $T_3=n+1$, the \rev{``Kari automaton''} (an automaton with $6$ states, see \cite{JKari01}) has $T_3=5$ and the \rev{``Roman automaton''} (an automaton with $5$ states, see \cite{Roman}) has $T_3=5$.}\footnote{Please note that it is possible to build automata with quadratic reset threshold and low TRT, see \cite{Ananichev200730} for such examples.}(close to the number of states). As the triple rendezvous time is much easier to compute than the shortest synchronizing word, it can be used as a heuristic filter in order to generate automata with large reset threshold.}

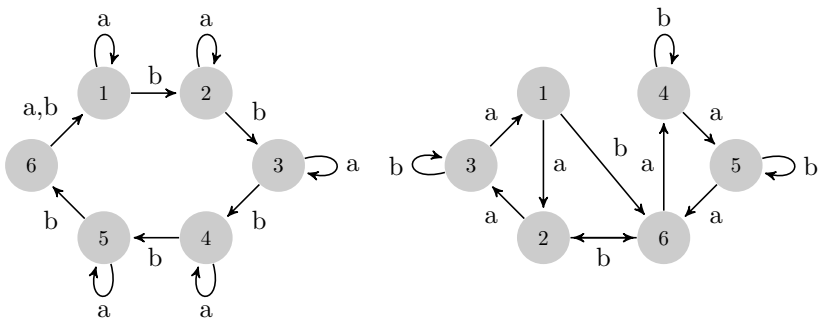
\begin{figure}[h!]
\begin{center}
\scalebox{1}{
\begin{tikzpicture}[->,>=stealth',shorten >=1pt,auto,node distance=1.7cm,
                    semithick]
  \tikzstyle{every state}=[fill=light-gray,draw=none,text=black, scale=0.8]

  \node[state] 			(A)                    {$1$};
  \node[state]         	(B) [right of=A] {$2$};
  \node[state]         	(C) [below right of=B] {$3$};
  \node[state]         	(D) [below left of=C] {$4$};
  \node[state]         	(E) [left of=D] {$5$};
  \node[state]         	(F) [below left of=A] {$6$};

  \path (A) edge              	node {b} (B)
  			edge [loop above]   node {a} (B)
        (B) edge              	node {b} (C)
            edge [loop above]   node {a} (B)
        (C) edge              	node {b} (D)
         	edge [loop right]   node {a} (B)
        (D) edge  			  	node {b} (E)
            edge [loop below]   node {a} (B)
        (E) edge  			  	node {b} (F)
            edge [loop below]   node {a} (B)
        (F) edge  			  	node {a,b} (A);
        
  \node[state]         	(L) [right of=C, node distance=3.2cm] {$3$};      
  \node[state] 			(G) [above right of=L] {$1$};
  \node[state]         	(H) [right of=G, node distance=2cm] {$4$};
  \node[state]         	(I) [below right of=H] {$5$};
  \node[state]         	(J) [below left of=I] {$6$};
  \node[state]         	(K) [left of=J, node distance=2cm] {$2$};

  \path (G) edge              	node {a} (K)
  			edge     			node {b} (J)
        (H) edge              	node {a} (I)
            edge [loop above]   node {b} (B)
        (I) edge              	node {a} (J)
         	edge [loop right]	node {b} (J)
        (J) edge  			  	node {b} (K)
            edge 				node {a} (H)
        (K) edge  			  	node {a} (L)
            edge    			node {} (J)
        (L) edge  			  	node {a} (G)
        	edge [loop left]	node {b} (I);

\end{tikzpicture}}
\end{center}
\caption{On the left the automaton of the \v Cern{\'y} family with 6 states, on the right \rev{the ``Kari automaton''}.}
\label{KariCzerny}
\end{figure}

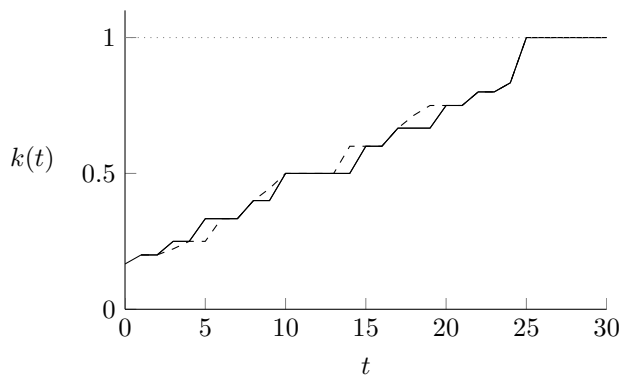
\begin{figure}[h!]
\begin{center}
\scalebox{1}{ 
\begin{tikzpicture}

\begin{axis}[%
width=2.52083333333333in,
height=1.565625in,
scale only axis,
xmin=0,
xmax=30,
xlabel={$t$},
ymin=0,
ymax=1.1,
ylabel={$k(t)$},
ylabel style={rotate=-90},
axis x line*=bottom,
axis y line*=left
]
\addplot [color=black,solid,forget plot]
  table[row sep=crcr]{%
1	0.200000000009069\\
2	0.199999999839946\\
3	0.250000003494876\\
4	0.250000000021004\\
5	0.333333333344143\\
6	0.333333333333414\\
7	0.333333333332746\\
8	0.400000000030815\\
9	0.40000000003424\\
10	0.500000002420791\\
11	0.500000000006423\\
12	0.499999999869416\\
13	0.499999999972829\\
14	0.499999999020076\\
15	0.599999999973505\\
16	0.600000003878179\\
17	0.666666666564197\\
18	0.66666666897568\\
19	0.666666666666799\\
20	0.749999999975103\\
21	0.750000000003013\\
22	0.800000000039944\\
23	0.800000000143612\\
24	0.833333335609524\\
25	1.00000000000172\\
26	0.999999999964388\\
27	0.999999999964388\\
28	0.999999999964388\\
29	0.999999999964388\\
30	0.999999999964388\\
};
\addplot [color=black,dashed,forget plot]
  table[row sep=crcr]{%
1	0.200000000000017\\
2	0.200000000451055\\
3	0.22222222098344\\
4	0.249999999984922\\
5	0.249999999978911\\
6	0.333333333279143\\
7	0.333333333333627\\
8	0.399999999997362\\
9	0.444444444445026\\
10	0.500000000018531\\
11	0.500000000606804\\
12	0.500000000011596\\
13	0.500000001982784\\
14	0.599999999998715\\
15	0.599999999935363\\
16	0.600000000114164\\
17	0.666666666668618\\
18	0.714285714290725\\
19	0.750000000000227\\
20	0.750000000000483\\
21	0.750000000038796\\
22	0.800000000000836\\
23	0.799999999739413\\
24	0.833333332880443\\
25	1.00000000005194\\
26	0.999999999964388\\
27	0.999999999964388\\
28	0.999999999964388\\
29	0.999999999964388\\
30	0.999999999964388\\
};
\addplot [color=black,solid,forget plot]
  table[row sep=crcr]{%
1	0.200000000009069\\
2	0.199999999839946\\
3	0.250000003494876\\
4	0.250000000021004\\
5	0.333333333344143\\
6	0.333333333333414\\
7	0.333333333332746\\
8	0.400000000030815\\
9	0.40000000003424\\
10	0.500000002420791\\
11	0.500000000006423\\
12	0.499999999869416\\
13	0.499999999972829\\
14	0.499999999020076\\
15	0.599999999973505\\
16	0.600000003878179\\
17	0.666666666564197\\
18	0.66666666897568\\
19	0.666666666666799\\
20	0.749999999975103\\
21	0.750000000003013\\
22	0.800000000039944\\
23	0.800000000143612\\
24	0.833333335609524\\
25	1.00000000000172\\
26	0.999999999964388\\
27	0.999999999964388\\
28	0.999999999964388\\
29	0.999999999964388\\
30	0.999999999964388\\
};
\addplot [color=black,solid,forget plot]
  table[row sep=crcr]{%
0	0.166666666666667\\
1	0.2\\
};
\addplot [color=gray,dotted,forget plot]
  table[row sep=crcr]{%
0	1\\
30	1\\
};
\end{axis}
\end{tikzpicture}}%
\end{center}
\caption{Synchronizing probability function for the automaton of the \v Cern{\'y} family with 6 states (solid curve), and for \rev{the ``Kari automaton''} (dashed curve).}
\label{KFuncKariCzerny}
\end{figure}

Figure \ref{KariCzerny} shows the automaton of the \v Cern{\'y} family with 6 states and \rev{the ``Kari automaton''} \cite{JKari01}, two automata achieving the bound of Conjecture \ref{cernyconj}, with synchronizing words of length 25. Figure \ref{KFuncKariCzerny} shows their SPF. For these automata, the \rev{TRT} is equal to 7 and 5 respectively. 
 
The following conjecture has been made on the behaviour of the SPF: 

\begin{conj} [Conjecture 2 in \cite{jungers_sync_12}]\label{ConjK}
In a synchronizing automaton $\Sigma$ with $n$ states, for any $1\leq j\leq n-1$, $$ k_{\Sigma}(1+(j-1)(n+1))\geq j/(n-1).$$
\end{conj}
This conjecture is stronger than \v Cern{\'y}'s conjecture (Theorem 4 in \cite{jungers_sync_12}). Conjecture \ref{ConjK} would also imply (see \cite{jungers_sync_12}) that the following conjecture about the triple rendezvous time is true:

\begin{conj} [Conjecture 4 in \cite{jungers_sync_12}]\label{ConjT3}
In a synchronizing automaton $\Sigma$ with $n$ states, 
$$T_{3,\Sigma}\leq n+2.$$
\end{conj}
In Section 4, we provide a family of automata which are counterexamples for both Conjecture \ref{ConjK} and Conjecture \ref{ConjT3}. 

We now focus on bounding the \rev{TRT}. A first upper bound can be easily obtained without using the SPF:

%\vspace{0.5 cm}

\begin{proposition}
In a synchronizing automaton $\Sigma$ with $n$ states, 
$$T_{3,\Sigma}\leq \frac{n(n-1)}{2}+1.$$
\label{InitialT3bound}
\end{proposition}
\begin{proof}
 Firstly, there are only  $n+n(n-1)/2$ possible different columns of weight one or two in $A$. Secondly, for any positive integer $t$ smaller than the \rev{reset threshold} (therefore also smaller than $T_3$), $A(t+1)$ must contain columns that are not in $A(t)$ (Lemma 1 in \cite{jungers_sync_12}). Therefore, as $A(0)$ includes the $n$ columns of weight one, $A(n(n-1)/2+1)$ includes at least $n+n(n-1)/2+1$ columns. As all the columns in $A$ are different, it implies that $A(n(n-1)/2+1)$ contains at least a column of weight superior or equal to 3.%\qed
\end{proof}

In order to obtain a better upper bound on \rev{the TRT}, we study the evolution of the SPF and $A(t)$ \rev{with respect to $t$}, for $t<T_3$. In that situation, $A(t)$ only contains columns of weight one or two. In the following, we name \emph{subset of columns} of a matrix $A \in \mathbb R^{n\times m}$ any matrix $A' \in \mathbb R^{n\times m'}$  (for some $m'\leq m$) obtained by erasing columns from $A$. A subset of columns of $A$ with some property is the subset of columns obtained by erasing the columns of $A$ not satisfying that property, and keeping all the others.

\begin{definition}We associate with $A(t)$ a graph $G(t)$ with $n$ vertices such that the subset of columns of weight two of $A(t)$ is the incidence matrix of $G(t)$.\end{definition}

\begin{example}
For the automaton from Example \ref{A3}, at $t=3$, the graph $G(3)$ associated with $A(3)$ is shown in Fig. \ref{GraphA3}.
\begin{figure}[h!]
\begin{center}
\scalebox{1}{
\begin{tikzpicture}[-,>=stealth',shorten >=1pt,auto,node distance=1.7cm,
                    semithick]
  \tikzstyle{every state}=[fill=light-gray,draw=none,text=black, scale=0.8]

  \node[state] (A)                    {$1$};
  \node[state]         (B) [above right of=A] {$2$};
  \node[state]         (D) [below right of=A] {$4$};
  \node[state]         (C) [below right of=B] {$3$};

  \path (B) edge              node {} (C)
        (C) edge              node {} (D)
        (D) edge  			  node {} (A);
\end{tikzpicture}}
\end{center}
\caption{Graph $G(3)$ associated with $A(3)$ for the automaton from Example \ref{A3}.}
\label{GraphA3}
\end{figure}
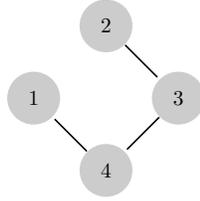
\demo
\end{example}

In the graph $G(t)$ associated with $A(t)$, we call a \emph{singleton} a vertex which is disconnected from the rest of the graph, a \emph{pair} two vertices which are connected to each other and disconnected from the rest of the graph, and a \emph{cycle} a set of vertices forming a cycle\footnote{$c$ vertices are forming a \emph{cycle} if we can number them from 1 to $c$ in such a way that vertex 1 is only connected to vertices 2 and $c$, each vertex $1<i<c$ is only connected to vertices $i-1$ and $i+1$, and the vertex $c$ is only connected to vertices $c-1$ and 1.} and disconnected from the rest of the graph. We call a cycle \emph{odd} (resp. \emph{even}) if it contains an odd (resp. even) number of vertices.

We also use the reverse correspondence.
Given a graph $G$ on $n$ vertices, each vertex $v_j$ is represented by $e_j^T,$ and an edge $(v_i,v_j)$ is represented by the vector $(e_i+e_j)^T.$
We notice that, \rev{if $1<t<T_3$, since $A(t)$ contains at least one column that is not in $A(t-1)$, $G(t)$ has at least one more edge than $G(t-1)$}.

Our general strategy to bound the \rev{TRT} is the following: first, based on this matrix-graph approach, we study the values that $k(t)$ can take with $t<T_3$. Then, we bound the maximal dimension that the polytope $P_t$ of solutions of Program (\ref{eqn-catchingprobalin}) can take for each value $k(t)$. We conclude by using the fact that this maximum strictly decreases if $k(t)$ stays constant.

To do so, we start from the set of reachable vectors $A(t)$. We prove that it is possible to find a subset of columns $A_c(t)$ of $A(t)$ satisfying the following properties; the matrix $A_c(t)$ is such that its associated graph $G_c(t)$ is composed of disjoint singletons, pairs and odd cycles, and is such that the optimal objective value for Program (\ref{eqn-catchingprobalin}) and Program (\ref{eqn-catchingprobalindual}) is the same if $A(t)$ is replaced by $A_c(t)$. This structure allows us to easily compute the value $k(t)$ and the dimension of the corresponding polytope of optimal solutions. Note that when replacing $A(t)$ with $A_c(t)$, the dimension of $P_t$ is non-decreasing with respect to $t$, as every optimal solution of Program (\ref{eqn-catchingprobalin}) with $A(t)$ is also an optimal solution with $A_c(t)$.

We call \emph{support} of the strategy $q$ of Program (\ref{eqn-catchingprobalindual}) the set of columns in $A(t)$ corresponding to non-zero entries in $q$. We call a column of $A(t)$ \emph{critical} if there exists an optimal solution $q$ with a non-zero entry corresponding to this column. 

In the proof of Lemma \ref{thm-npluss} hereunder, we split the program into two subprograms and introduce \rev{notation} to get to our arguments. The reader may refer to Example \ref{exampleThm} to have an illustration of the steps.

\begin{lemma}
If $t<T_3$, there exists an optimal solution q for Program (\ref{eqn-catchingprobalindual}) such that its support is associated with a graph composed of disjoint singletons, pairs, and odd cycles.
\label{thm-npluss} 
\end{lemma}

\begin{proof} We first present the structure of the proof, then we demonstrate the technical elements.

\textbf{Structure of the proof.}
We proceed by induction on the number of rows of matrix $A(t)$, which is also the number of vertices of $G(t)$. If there is only one or two vertices, the lemma is trivially true. Otherwise, we use the fact that the graph associated with a subset of columns of $A(t)$ can either be connected or disconnected. 

If it is disconnected, we can define two subprograms with the structure of Program (\ref{eqn-catchingprobalindual}) with less variables, for which we know by induction that there exist optimal solutions satisfying the lemma. From these solutions, we can build an optimal solution of the original program whose support \rev{is also} composed of disjoint singletons, pairs, and odd cycles. 

If it is connected, we can either find an optimal solution with a support associated with an odd cycle including all the vertices, or find a solution with a support associated with a disconnected graph. In this latter case, we \rev{are again} able to split the program as in the disconnected case.

\textbf{Demonstration.} 
Let us write $Prog_A$ for \rev{the original} Program (\ref{eqn-catchingprobalindual}), and $Prog_N$ for Program (\ref{eqn-catchingprobalindual}) \rev{in which} matrix $A$ \rev{is} replaced by a matrix $N$, with the corresponding dimensions for vectors $q$ and $e$. Let us write $k_A$ and $k_N$ the optimal objective value of $Prog_A$ and $Prog_N$.

We suppose by induction that the Lemma holds if $A(t)$ is associated with a graph with $n-1$ vertices or less (so if $A(t)$ has $n-1$ rows or less). Let us now consider a matrix $A(t)$ associated with a graph $G(t)$ with $n$ vertices. Let $A_{min}$ be a \emph{minimal subset} of columns of $A(t)$ (that is, such that $k_A=k_{A_{min}}$, and $k_A'<k_A$ for any strict subset of columns $A'$ of $A_{min}$) and with not more than $n$ columns (such a matrix exists, see \cite[Proposition 4]{jungers_sync_12}). Let $G_{min}$ be the graph associated with $A_{min}$. We will prove that either $G_{min}$ is already composed of disjoint singletons, pairs and odd cycles, or that we can build an optimal solution $q_c$ to $Prog_{A_{min}}$ such that the graph associated with its support is composed of disjoint singletons, pairs and odd cycles. 
Since the program $Prog_{A_{min}}$ has the same optimal objective value as $Prog_A$, there is a vector $q$ which is an optimal solution of $Prog_{A_{min}}$ (such that $A_{min}q\geq k_Ae_{n\times 1}).$

Two cases can occur: either $G_{min}$ is connected, or it is disconnected. 

\textbf{Disconnected case.} 
If $G_{min}$ is disconnected, then the graph can be split into two separate graphs $G_{min1}$, $G_{min2}$ with respectively $n_1$ and $n_2$ vertices. For these two graphs, let us consider the associated matrices $A_{min1}$ and $A_{min2}$ (with $n_1$ and $n_2$ rows respectively). We will call $Prog_{A_{min1}}$ and $Prog_{A_{min2}}$ \emph{subprograms} of $Prog_{A_{min}}$.
\newline
Now define $q_1$ and $q_2$ the subvectors of $q$ corresponding to the graphs $G_{min1}$ and $G_{min2}$ respectively. That is, the entries of $q_1$ are the same as the entries of $q$ corresponding to columns from $Prog_{A_{min1}}$, and the entries of $q_2$ are the same as the entries of $q$ corresponding to columns from $Prog_{A_{min2}}$. As $q$ is optimal, we have that $A_{min1} q_1 \geq k_A e_{n_1 \times 1}$ and $A_{min2} q_2 \geq k_A e_{n_2\times 1}$ (these inequalities and the following ones are entrywise).
\newline
Then define $w_1=1/(eq_1),$ $w_2=1/(eq_2),$ the inverse of the weight associated with each subprogram in the strategy $q$. 
\newline
We have that $q_1w_1$ and $q_2w_2$ are feasible solutions for $Prog_{A_{min1}}$ and for $Prog_{A_{min2}}$ respectively, so $ A_{min1} q_1 w_1\geq w_1k_Ae_{n_1 \times 1}$ and $ A_{min2} q_2 w_2 \geq w_2k_Ae_{n_2 \times 1}$. Therefore, $k_{A_{min1}}\geq w_1k_A$ and $k_{A_{min2}}\geq w_2k_A$. 
\newline
By the induction hypothesis, for $Prog_{A_{min1}}$ and for $Prog_{A_{min2}}$, there are vectors $r_1$ and $r_2$ which are optimal solutions and have a support associated with a graph composed of singletons, pairs and odd cycles. Therefore for these vectors we have
\begin{eqnarray} \nonumber
 A_{min1}r_1\geq k_{A_{min1}}e_{n_1 \times 1} \geq w_1k_Ae_{n_1 \times 1}\\\nonumber
A_{min2}r_2\geq k_{A_{min2}} e_{n_2 \times 1}\geq w_2k_Ae_{n_2 \times 1}.
\end{eqnarray}
Let us now extend $r_1$ and $r_2$ to vectors $q_{c1}$ and $q_{c2}$ of the same length as $q$, with the entries corresponding to the related subprogram equal to the ones of $r_1$ and $r_2$, and the other entries equal to zero. For these strategies, we have that $Aq_1'\geq w_1k_Ae$ and $Aq_2'\geq w_2k_Ae$ from the argument above. Now defining $q_c=q_{c1}/w_1+q_{c2}/w_2,$ we have $e^Tq_c=1$ and $Aq_c\geq k_A,$ which implies that $q_c$ is an admissible optimal solution for $Prog_A$.
\newline
Since the union of two graphs composed of disjoint singletons, pairs and odd cycles is also composed of disjoint singletons, pairs and odd cycles, the graph \rev{$G_c$} associated with the support of the strategy $q_c$ is composed of disjoint singletons, pairs and odd cycles, as wanted.

\textbf{Connected case.}
For the connected case, if $G_{min}$ is connected and $n\geq 3$, we can either show that $G_{min}$ has no vertex of degree one, or build an optimal solution such that its support is disconnected, which leads to the previous case.
\newline
Indeed, consider an edge $(v_1,v_2)$ with the vertex $v_2$ of degree one.  If $v_1$ is of degree one, then $(v_1,v_2)$ is a disconnected pair and $G_{min}$ is not connected. If $v_1$ has \rev{other} adjacent edges $(v_1,v_i)$, with $i>2$, one can obtain a new solution in which the value corresponding to \rev{these other edges} is equal to zero, without changing $k_A$ and without adding value to edges that are not in the support of the current solution.

\rev{For the description of this construction, let us denote by $q_{v_\alpha}$ the entry of $q$ corresponding to the vertex $v_\alpha$, and by $q_{(v_\beta,v_\gamma)}$ the entry of $q$ corresponding to the edge $(v_\beta, v_\gamma)$, with $1\leq \alpha \leq n$ and $1\leq \beta \neq \gamma\leq n$. From the solution $q$, we build a vector $q'$ of the same dimension as follows:
\begin{align}
\label{defq}
q'_{v_1}= q'_{v_2}&=0,&\\\nonumber
q'_{(v_1,v_2)}&=q_{v_1}+q_{v_2}+q_{(v_1,v_2)},&\\\nonumber
q'_{(v_1,v_i)}&=0& \text{for } i>2 \text{ and } (v_1,v_i)\in G_{min},\\\nonumber
q'_{v_i}&=q_{(v_1,v_i)}+q_{v_i}& \text{for } i>2 \text{ and } (v_1,v_i)\in G_{min},\nonumber
\end{align}
and all the other entries of $q'$ equal to the entries of $q$.

With this construction, we have the following inequalities: $$q'_{v_1}+\sum_{j\neq 1,\text{ } (v_1,v_j)\in G_{min}} q'_{(v_1,v_j)}=q_{v_1}+q_{v_2}+q_{(v_1,v_2)}\geq k_A,$$
$$q'_{v_2}+\sum_{j\neq 2,\text{ } (v_2,v_j)\in G_{min}} q'_{(v_2,v_j)}=q_{v_1}+q_{v_2}+q_{(v_1,v_2)}\geq k_A,$$ as $q_{v_2}+q_{(v_1,v_2)}\geq k_A$, because $q$ is an optimal solution. Moreover, for every other $i>2$, $$q'_{v_i}+\sum_{j\neq i, \text{ } (v_i,v_j)\in G_{min}} q'_{(v_i,v_j)}=q_{v_i}+\sum_{j\neq i,\text{ } (v_i,v_j)\in G_{min}} q_{(v_i,v_j)}\geq k_A.$$ 
In addition, it is straightforward from (\ref{defq}) that $\sum_{i=1}^{\rev{m(t)}} q_i'=\sum_{i=1}^{\rev{m(t)}} q_i=1 $. Therefore $q'$ and $k_A$ satisfy the conditions of Program (\ref{eqn-catchingprobalindual}), and $q'$ is an optimal strategy.
\newline
The value associated with the edges $(v_1,v_i)$, $i>2$, in $q'$ is $0$, so these edges are not in the support of $q'$, and the pair $(v_1,v_2)$ is disconnected. 
Therefore either $G_{min}$ has no vertex of degree one, or we can build a solution with a disconnected support.}
\newline
Now, notice that a connected graph with not more than $n$ edges and no vertex of degree one cannot have vertices of degree larger than two (because the sum of the degrees is equal to twice the number of edges). Therefore, all the vertices are in the same cycle. \rev{If $n$ is odd, we have our result. If $n$ is even, let us label its vertices $v_1, \dots, v_n$ in the cyclic order and define the vector $q'$ of the same dimension as $q$ as follows:
\begin{align}
\label{defqbis}\nonumber
q'_{v_i}&=0,& \text{for }1\leq i \leq n,\\\nonumber
q'_{(v_{2j-1},v_{2j})}&=q_{v_{2j-1}}+q_{v_{2j}}
+q_{(v_{2j-1},v_{2j})}+q_{(v_{2j},v_{2j+1})}& \text{for }1\leq j \leq n/2-1,\\\nonumber
q'_{(v_{n-1},v_n)}&=q_{v_{n-1}}+q_{v_n}+q_{(v_{n-1},v_n)}
+q_{(v_{n},v_{1})},& \\\nonumber
q'_{(v_{2l},v_{2l+1})}&=0& \text{for } 1\leq l \leq n/2-1,\\\nonumber
q'_{(v_{n},v_1)}&=0.&
\end{align}
We can verify that $q'$ is a feasible solution reaching the objective value $k_A$. Since the edges $(v_{2l},v_{2l+1})$ for $1\leq l \leq n/2$ and the edge $(v_{n},v_{1})$ are not in its support, it leads again to a disconnected case, and the proof is done}.%\qed
\end{proof}

\begin{example}
\label{exampleThm}
The purpose of this example is to illustrate the steps of Lemma \ref{thm-npluss} in the disconnected case. 

We start from Example \ref{A3}. We have
\rev{ $$A(3)=\left( \begin{array}{ccccccc}
1 & 0 & 0 & 0 &  1 & 0 & 0\\[-0.05cm]
0 & 1 & 0 & 0 &  0 & 0 & 1\\[-0.05cm]
0 & 0 & 1 & 0 &  0 & 1 & 1\\[-0.05cm]
0 & 0 & 0 & 1 &  1 & 1 & 0\end{array} \right).$$}
A possible minimal subset of columns of $A(3)$ is the following: 
\rev{$$A_{min}=\left( \begin{array}{cc}
 1 & 0\\[-0.05cm]
 0 & 1\\[-0.05cm]
0 & 1\\[-0.05cm]
1 & 0 \end{array} \right),$$}
which is the support of the solution of the original program $q=(0, 0, 0, 0, 1/2, 0, 1/2)^T$.

The associated graph $G_{min}$ is represented in Fig. \ref{ExampleG'}. It is composed of two pairs. It already is composed of singletons, pairs, and odd cycles, however we will continue the analysis to illustrate the ideas.

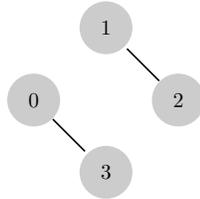
\begin{figure}[h!]
\begin{center}
\scalebox{1}{
\begin{tikzpicture}[-,>=stealth',shorten >=1pt,auto,node distance=1.7cm,
                    semithick]
  \tikzstyle{every state}=[fill=light-gray,draw=none,text=black, scale=0.8]

  \node[state] (A)                    {$0$};
  \node[state]         (B) [above right of=A] {$1$};
  \node[state]         (D) [below right of=A] {$3$};
  \node[state]         (C) [below right of=B] {$2$};

  \path (A) edge              node {} (D)
        (C) edge              node {} (B);
\end{tikzpicture}}
\end{center}
\caption{Graph $G_{min}$ associated with $A_{min}$.}
\label{ExampleG'}
\end{figure}

As the graph is disconnected, we separate it into $G_{min1}$ and $G_{min2}$, being the subgraphs induced by respectively $\{0,3\}$ and $\{1,2\}$.
This gives us the incidence matrices $A_{min1}=(1, 1)^T$ and $A_{min2}=(1, 1)^T$. We now consider the programs $Prog_{A_{min1}}$ and $Prog_{A_{min2}}$. Optimal solutions of these subprograms are $r_1=(1)$ and $r_2=(1)$, which each have a support forming a single pair (in a general setting we have to use the induction argument to prove the existence of solutions with support associated to a graph composed of singletons, pairs and odd cycles).

From the strategy $q$, we obtain $w_1=2$ for $Prog_{A_{min1}}$ and $w_2=2$ for $Prog_{A_{min2}}$. Expanding these solutions gives us $q_{c1}=(1,0)^T$, $q_{c2}=(0,1)^T$. This leads to the solution of the original program $q_c=q_{c1}/w_1+ q_{c2}/w_2= (1/2, 1/2)^T$, which has a support composed of two pairs.
\demo
\end{example}

It turns out that, given a set $A(t)$, the knowledge of a subset of its columns which is the support of an optimal solution of Program (\ref{eqn-catchingprobalindual}) and is composed of disjoint odd cycles, pairs and singletons allows to easily compute the SPF $k(t)$. Moreover, it provides with an upper bound on the dimension of the set of solutions $P_t$:

\begin{lemma}
\label{LemmaDim}
If the graph $G(t)$ associated with $A(t)$ is composed of disjoint odd cycles, pairs and singletons, then the optimum of Program (\ref{eqn-catchingprobalindual}) is given by $2/(n+n_1),$ where $n_1$ is the number of singletons.
Moreover, the dimension of $P_t$ is the number of pairs.
\end{lemma}

\begin{proof}
To make notations concise, we define $K=n+n_1.$  Our claim is that $k(t)=2/K.$ We provide an admissible solution for the primal Program (\ref{eqn-catchingprobalin}), as well as for the dual Program (\ref{eqn-catchingprobalindual}), with the same objective value. Therefore this value is optimal. 

A solution of Program (\ref{eqn-catchingprobalin}) can be built as follows:
\begin{equation}
\label{S1}
p_i=
\left\lbrace
\begin{array}{ll}
2/K  & \mbox{if state $i$ corresponds to a singleton vertex,}\\
1/K & \mbox{otherwise.} 
\end{array}\right.
\end{equation}
The sum of the coefficients is $\sum_{i=1}^np_i=(n-n_1)/K+2n_1/K=1,$ and $p$ is a feasible solution for Program (\ref{eqn-catchingprobalin}) with an objective value of $2/K$.

For Program (\ref{eqn-catchingprobalindual}), a solution can be built as follows:
\begin{equation}
\label{S2}
q_i=
\left\lbrace
\begin{array}{ll}
2/K  & \mbox{if column $i$ corresponds to a pair,}\\
1/K  & \mbox{if column $i$ corresponds to an edge of an odd cycle,}\\
2/K & \mbox{if column $i$ corresponds to a singleton.} 
\end{array}\right.
\end{equation}
The sum of the coefficients is $\sum_{i=1}^{m(t)}q_i=n_1\times 2/K+(n-n_1) \times 1/K=1,$ and $q$ is a feasible solution for Program (\ref{eqn-catchingprobalindual}) with objective value of $2/K$.
Summarizing, Equation (\ref{S1}) describes a solution for Program (\ref{eqn-catchingprobalin}), and Equation (\ref{S2}) describes a solution for Program (\ref{eqn-catchingprobalindual}), achieving the same objective value. As the programs are dual, this implies that both strategies are optimal, and $k(t)=2/K.$

We can now give an explicit expression for $P_t$.
Reordering the vertices such that the first $f$ indices correspond to singletons, the next $g$ indices correspond to vertices in pairs, grouped by pair (two vertices in the same pair have indices $f+2j-1$ and $f+2j$, $1\leq j \leq g/2$), and the last $h$ indices correspond to vertices in odd cycles, we have the following set of optimal solutions for Program (\ref{eqn-catchingprobalin}):
\begin{equation}
\label{polytope}
P_t=
\left\lbrace\begin{array}{l|ll}
& p_i=2/K ,& \text{with } 1\leq i \leq f,\\ 
 & p_{f+2j-1}= 1/K+x_j, & \\ 
 (p_1, ..., p_{f+g+h} )& p_{f+2j}=1/K-x_j, &\text{with }1\leq j \leq g/2 \text{ and } \\
 && -1/K\leq x_j \leq 1/K,\\ 
 & p_k=1/K, &  \text{with } f+g+1\leq k \leq f+g+h 
\end{array}\right\rbrace.
\end{equation}
One can check that the vectors $p$ described in (\ref{polytope}) are feasible solutions for Program (\ref{eqn-catchingprobalin}). Each element of the set $P_t$ is uniquely defined by $g/2$ independent parameters $x_j$ that can vary between $-1/K$ and $1/K$, so the dimension of $P_t$ is $g/2$.
\end{proof}

We now combine these lemmas to obtain a universal upper bound on the triple rendezvous time for synchronizing automata. The main steps of the reasoning are as follows:
\rev{starting from any original program obtained from an automaton and a value $t$, due to Lemma \ref{thm-npluss}, one can find a subset of columns $A_c(t)$ of $A(t)$ such that $Prog_{A_c(t)}$ has the same objective value as $Prog_{A(t)}$ and such that the associated graph satisfies the conditions of Lemma \ref{LemmaDim}. For this program, we can easily compute the value of the SPF and an upper bound on the dimension of $P_t$. Then, using dimensional arguments, we obtain a lower bound on the SPF growing rate with respect to $t$, for $t<T_3$:}

\begin{theorem} 
\label{cor-stagnate}
If $t<T_3$, then $k(t)$ can only take the values $2/(n+s),$ $0\leq s  \leq n-1$, and this value cannot be optimal at more than $\lfloor (n-s)/2 \rfloor+1$ consecutive values of $t$.
\end{theorem}
\begin{proof}
Let us fix $t<T_3$. By Lemma \ref{thm-npluss}, one can find a matrix $A_c$ which is a subset of the columns of $A(t)$, such that the associated graph $G_{c}$ is composed of singletons, pairs and odd cycles, and such that $k_A(t)=k_{A_c}$. Let $R_t$ be the set of optimal solutions of $Prog_{A_{c}}$. From Lemma \ref{LemmaDim}, the dimension of $R_t$ is the number of pairs in $G_{c}$. Let $s$ be the number of singletons. There are $n-s$ vertices of $G_{c}$ which are either in pairs or in odd cycles, so there are at most $(n-s)/2$ pairs. The set of optimal solutions $R_t$ of $Prog_{A_{c}}$ has a dimension not smaller than the set of optimal solutions $P_t$ of $Prog_{A}$, because optimal solutions of $Prog_{A}$ are also optimal solutions of $Prog_{A_{c}}$. Therefore we have that $\text{dim}(P_t)\leq \text{dim}(R_t)\leq(n-s)/2$. 

We claim that, if $k(t+1)=k(t)$, then $\text{dim}(R_{t+1})<\text{dim}(R_t)$.
Equation (\ref{S2}) gives the set of optimal solutions, in which a positive value is assigned to each of the columns of $A_c(t)$, so it is a set of critical columns. 
We will use a similar argument as in the proof of \cite[Theorem 3]{jungers_sync_12} to obtain that $\text{dim}(R_{t+1})<\text{dim}(R_t)$. 

Hence we suppose that $k(t+1)=k(t)$. First recall that, from Lemma \ref{inclusion}, $$P_{t+1}\subset P_t.$$
We define $$A'(t+1)=A_c(t)\cup\{MA_c(t):M\in \Sigma \}$$
and $P'_{t+1}$ as the set of solutions $p\geq 0, e^T p=1$, such that $p^TA'(t+1)=ke$. Note that $P_{t+1}\subset P_{t+1}'$, since the columns in $A'(t+1)$ are a subset of the critical columns in $A(t+1)$. Also, $P'_{t+1}\subset R(t)$ since the columns in $A_c(t)$ are a subset of columns in $A'(t+1)$.

We first show that $P_{t+1}'\neq R_t$. Indeed, since $A'(t+1)$ contains all the columns of $A_c(t)$ multiplied by a matrix in $\Sigma$, it is clear that 
$$\forall M\in \Sigma, \forall p \in P'_{t+1}, M^Tp\in R_t.$$
Supposing $P_{t+1}'= R_t$, the above equation implies that $B^TR_t\subset R_t$ for all $B\in \Sigma^s, s\geq 1$, which implies that $\Sigma$ is not synchronizing. Indeed, this implies that for all $p \in R_t$, for all $B\in \Sigma^s, p^TB\leq ke^T.$

So, there must be a matrix $M\in \Sigma$ and a column $a_j$ of $A_c(t)$ such that $Ma_j\notin A_c(t)$. Again, since $MA_c(t)q\geq ke, Ma_j$ is a new critical column.

Now, by Theorem \ref{THM2}, the new critical column $Ma_j$ is such that $p^TMa_j=k$ for all $p\in P_{t+1}'$. Let $H$ be the hyperplane represented by this constraint. Since $R_t \cap H \neq R_t$, and $R_{t+1} \subset R_t\cap H$, it follows that $\text{dim}(R_{t+1})<\text{dim}(R_t)$, which proves the claim.

Therefore, as $\text{dim}(R_t)\leq \lfloor(n-s)/2\rfloor$, if $k(t)=k(t+\lfloor(n-s)/2\rfloor+1)$, the dimension of $(R_{t+\lfloor(n-s)/2\rfloor+1})$ would have to be negative, which is impossible. This implies that  $k(t+\lfloor(n-s)/2\rfloor+1)>k(t)$.
\end{proof}

With Theorem \ref{cor-stagnate}, we can now obtain the bound:

\begin{theorem}
In a synchronizing automaton $\Sigma$ with $n$ states, $$T_{3,\Sigma}\leq \frac{n(n+4)}{4} - \frac{n \text{ mod } 2}{4}.$$
\label{FirstT3bound}
\end{theorem}
\begin{proof}
By Theorem \ref{cor-stagnate}, if $t<T_3$, $k(t)$ can only take the values $2/(n+s),$ with $0\leq s \leq n-1$. Moreover, if $k(t)=2/(n+s)$, then $k(t+\lfloor(n-s)/2\rfloor+1)>k(t)$.  Summing over all possible values for $k(t)$ if $t<T_3$, one gets
\begin{equation}
\label{firstbound}
\sum_{s=0}^{n-1} \left(\lfloor(n-s)/2\rfloor+1 \right)=\sum_{s=1}^{n}\left(  \lfloor s/2\rfloor+1 \right) 
= \frac{n(n+4)}{4} - \frac{n \text{ mod } 2}{4}. 
\end{equation} %\qed 
\end{proof}

We now present another technique that provides with an alternative upper bound on the triple rendezvous time. It is based on the observation that, if the synchronizing probability function at $t=n$ is small, the \rev{TRT} is also small. Then, coupling \rev{this with Theorem \ref{cor-stagnate}}, we will obtain a further improvement of the bound.

\begin{lemma}
For  $1\leq s\leq n/2,$ if $k(n)<\frac{1}{n-s}$, then for any word $W\in \Sigma^{\leq n}$, there are less than $s$ zero entries in $eW$.
\label{SubLemmeComplet}
 \end{lemma}
\begin{proof} Suppose that there is a word $L\in \Sigma^{\leq n}$ such that $eL$ has at least $s$ zero entries, and therefore at most $n-s$ non-zero entries. Let $p$ be an optimal solution for Program (\ref{eqn-catchingprobalin}). Applying the word $L$ to the automaton, the final probability distribution on the states is $pL$. Entrywise, $pL\leq eL$, therefore $pL$ has at most $n-s$ non-zero entries. As the sum of the entries of $pL$ is equal to one, the average of its non-zero entries is at least $1/(n-s)$. Therefore, at least one of the entries is higher than or equal to $1/(n-s)$. As $p$ is an optimal solution, it implies that $k(n)\geq\frac{1}{n-s}$. Therefore we have the contrapositive: if $k(n)<\frac{1}{n-s}$, for any word $W\in \Sigma^{\leq n}$ there are less than $s$ zero entries in $eW$.%\qed
\end{proof}

Now, based on Lemma \ref{SubLemmeComplet}, we have the following result on the triple rendezvous time:

\begin{lemma}
For any strongly connected synchronizing automaton with $n$ states, and for any integer $1\leq s\leq n/2$, either
$$k(n)\geq \frac{1}{n-s} $$
or 
$T_3\leq n(s+2)/2. $
\label{LemmeComplet}
 \end{lemma}
	\begin{proof} Let us fix $s$ as in the statement of Lemma \ref{SubLemmeComplet}. We suppose that $k(n)< \frac{1}{n-s}$, and our goal is to prove that $T_3\leq n(s+2)/2.$ Let $t_s$ be the minimal $t$ such that $G(t_s)$ has a vertex of degree larger than $s$, and let $v$ be such a vertex. As there is at least one more edge in $G(t+1)$ than in $G(t)$, and there are $n$ vertices, we have that $t_s\leq sn/2$. This means that there exists a state $q$ corresponding to $v$, $s$ states $q_i$ corresponding to the vertices connected to $v$ in $A(sn/2)$, and words $W_i\in \Sigma^{\leq sn/2}$, with $1\leq i\leq s$, such that $qW_i=q_iW_i$.

We will now build a word with a column of weight three. As the graph is strongly connected, starting with a letter having a column of weight two, there exists a word $W_1$ of length at most $n$ and two states $q_1$ and $q_2$ such that $q_1W_1=q_2W_1=q$.
From Lemma \ref{SubLemmeComplet}, there are less than $s$ zeros in any product of length $n$. So,  there exists one state $q_3$ such that its corresponding vertex is connected to $v$ in $G(sn/2)$, and such that its corresponding entry in $eW_1$ is at least one. Let us call $W_2$ the word of length at most $sn/2$ such that $qW_2=q_3W_2$.
The word $W_1 W_2$ is such that three states are mapped to a single state. Therefore either $T_3\leq sn/2+n $ or $k(n)\geq \frac{1}{n-s} $.%\qed
\end{proof}

Coupling Theorem \ref{cor-stagnate} and Lemma \ref{LemmeComplet}, we obtain the following upper bound:

\begin{proposition}
For any strongly connected synchronizing automaton with $n$ states, for any $1\leq s\leq n/2,$
\begin{equation}\label{eq-bound-T3-max} T_3 \leq
\max{\{n(s+2)/2 , (n(n+4)-(2s-1)(2s+3)+1)/4\}} \end{equation}
\end{proposition}
\begin{proof}
If $$k(n)< \frac{1}{n-s},$$ we apply Lemma \ref{LemmeComplet}.

Otherwise, from Theorem \ref{cor-stagnate}, the values that can be taken by the synchronizing probability function for $t$ between $n$ and $T_3$ are of the shape $$k(t)=\frac{2}{n+r},$$ with $0\leq r \leq n-2s,$ and if $k(t)=2/(n+r)$, then $k(t+\lfloor(n-r)/2\rfloor+1)>k(t)$.

Summing over all possible values for $k(t)$, one gets
\begin{eqnarray} \nonumber \sum_{r=0}^{n-2s} \lfloor(n-r)/2\rfloor +1 &=&\sum_{r=2s}^{n}\lfloor(r)/2\rfloor +1 \\
\nonumber  &=&\frac{(n)(n+4)}{4} - \frac{n\text{ mod }2}{4}-\frac{(2s-1)(2s+3)}{4}+\frac{1}{4}\\
\nonumber &\leq&(n(n+4)-(2s-1)(2s+3)+1)/4.  \end{eqnarray}%\qed
\end{proof}

This inequality on $T_3$ holds for all $s$. Therefore, minimizing the bound in (\ref{eq-bound-T3-max}) over s, we obtain:

\begin{theorem}
In a strongly connected synchronizing automaton $\Sigma$ with $n$ states, $$T_{3,\Sigma}\leq (\sqrt{5n^2+4n-12}-n+6)n/8. $$
\label{HighT3bound}
\end{theorem}
\begin{proof}
In Equation (\ref{eq-bound-T3-max}), $n(s+2)/2$ is an increasing function of $s$, while $(n(n+4)-(2s-1)(2s+3)+1)$ is a decreasing function of $s$. Therefore, the minimum of (\ref{eq-bound-T3-max}) is reached at the intersection of the two functions. The value of $s$ that minimizes the expression in (\ref{eq-bound-T3-max}) is the solution of 
$$n(s+2)/2 = (n(n+4)-(2s-1)(2s+3)+1)/4.$$
 This is equivalent to:
$$s^2+s(n+2)/2 +(1-n^2/4)=0, $$
which has the solution
$$s=(\sqrt{5n^2+4n-12}-(n+2))/4.$$
This solution is positive and lower than $n/2$ as long as $n>3$. Plugging this value into Equation (\ref{eq-bound-T3-max}) we obtain
$$T_3\leq n(\sqrt{5n^2+4n-12}-n+6)/8.$$
\end{proof}

Therefore, we have that $T_3<n^2/(6.4...)$ for $n$ sufficiently large. This bound strictly improves on our previous one (\ref{firstbound}).
\section{A Counterexample to a Conjecture on the Synchronizing Probability Function}

In this section, we present an infinite family of automata which are counterexamples to Conjecture \ref{ConjK} and Conjecture \ref{ConjT3}. This family provides us with a lower bound on the maximum value of the triple rendezvous time for automata with $n$ states for every odd integer $n\geq 9$. Let us name the automaton of the family with $n$ states $\mathcal {TR}_n$. The automaton $\mathcal{TR}_9$ is shown in  Fig. \ref{Counter}. It has $9$ states, labelled from $q_1$ to $q_9$, and two letters $a$ and $b$ defined as (zeros are replaced by dots):
\arraycolsep=1.8pt\def\arraystretch{0.8}
$$a=\left( \begin{array}{ccccccccc}
.&.&.&.&.&.&1&.&.\\
.&.&.&1&.&.&.&.&.\\
.&.&1&.&.&.&.&.&.\\
.&1&.&.&.&.&.&.&.\\
.&.&1&.&.&.&.&.&.\\
.&.&.&.&.&.&.&1&.\\
1&.&.&.&.&.&.&.&.\\
.&.&.&.&.&1&.&.&.\\
.&.&.&.&.&.&.&.&1 \end{array} \right)
b=\left( \begin{array}{ccccccccc}
.&1&.&.&.&.&.&.&.\\
.&.&1&.&.&.&.&.&.\\
1&.&.&.&.&.&.&.&.\\
.&.&.&.&1&.&.&.&.\\
.&.&.&.&.&1&.&.&.\\
.&.&.&1&.&.&.&.&.\\
.&.&.&.&.&.&.&.&1\\
.&.&.&.&.&.&.&1&.\\
.&.&.&.&.&.&1&.&.\end{array} \right).$$% 
.

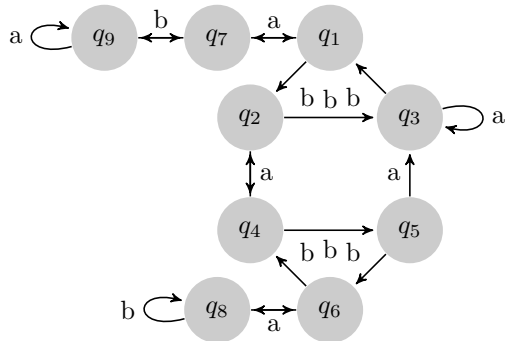
\begin{figure}[!htb]
\begin{center}
\scalebox{1}{
\begin{tikzpicture}[->,>=stealth',shorten >=1pt,auto,node distance=1.5cm,
                    semithick]
  \tikzstyle{every state}=[fill=light-gray,draw=none,text=black, scale=1]

  \node[state] 			(A)                    {$q_9$};
  \node[state]         	(B) [right of=A] {$q_7$};
  \node[state]         	(C) [ right of=B] {$q_1$};
  \node[state]         	(D) [below left of=C, node distance=1.5cm] {$q_2$};
  \node[state]         	(E) [below right of=C, node distance=1.5cm] {$q_3$};
  \node[state]         	(F) [below of=D] {$q_4$};
  \node[state]         	(G) [below of=E] {$q_5$};
  \node[state]         	(H) [below right of=F, node distance=1.5cm] {$q_6$};
  \node[state]         	(I) [left of=H] {$q_8$};

  \path (A) edge              	node {b} (B)
            edge [loop left]  	node {a} (A)
        (B) edge   				node {} (A)
            edge              	node {a} (C)
        (C) edge              	node {b} (D)
         	edge  				node {} (B)
        (D) edge  			  	node {b} (E)
            edge  			  	node {a} (F)
        (E) edge  			  	node {b} (C)
            edge [loop right]	node {a} (E)
        (F) edge  			  	node {} (D)
            edge 			[swap]	node {b} (G)
        (G) edge  			  	node {a} (E)
            edge 		[swap]		node {b} (H)
        (H) edge  			  	node {a} (I)
            edge 		[swap]		node {b} (F)
        (I) edge  			  	node {} (H)
            edge [loop left]	node {b} (I);
\end{tikzpicture}}
\end{center}
\caption{The automaton $\mathcal {TR}_9$.}
\label{Counter}
\end{figure}

\rev{With the parameter $j$ in Conjecture \ref{ConjK} fixed to $2$, this conjecture implies that for any synchronizing automaton with $9$ states, $k(11)\geq 2/8$. Conjecture \ref{ConjK} is stronger than Conjecture \ref{ConjT3}. This latter implies, for the same automaton, that $T_{3}\leq 11$.
However, we have that $T_{3,\mathcal{TR}_9}=12$, thus disproving both conjectures: 

\begin{proposition}
The TRT of the automaton $\mathcal {TR}_9$ is equal to 12.
\end{proposition}
\begin{proof}
On the one hand, the set of reachable \rev{vectors} at $t=11$, in which we separated the columns by blocks to show the evolution of $A(t)$ with respect to $t$, is equal to: 

\arraycolsep=1.8pt\def\arraystretch{0.8}
$$A_{tr_9}(11)=
\left( \begin{array}{ccccccccc|c|c|c|c|c|c|c|c|cc|ccc|ccc}
1 & . & . & . & . & . & . & . & . & . & . & 1 & . & . & . & . & . & . & . & 1 & 1 & . & . & . & . \\
. & 1 & . & . & . & . & . & . & . & . & 1 & . & . & . & . & . & . & 1 & . & . & . & . & . & . & . \\
. & . & 1 & . & . & . & . & . & . & 1 & . & . & . & . & . & . & . & . & . & . & . & . & 1 & 1 & . \\
. & . & . & 1 & . & . & . & . & . & . & 1 & . & . & . & . & . & 1 & . & . & . & . & . & . & . & 1 \\
. & . & . & . & 1 & . & . & . & . & 1 & . & . & . & . & . & 1 & . & . & . & . & . & 1 & . & . & . \\
. & . & . & . & . & 1 & . & . & . & . & . & 1 & . & . & 1 & . & . & . & 1 & . & . & . & . & . & . \\
. & . & . & . & . & . & 1 & . & . & . & . & . & 1 & . & . & 1 & . & . & 1 & 1 & . & . & . & . & 1 \\
. & . & . & . & . & . & . & 1 & . & . & . & . & 1 & 1 & . & . & . & . & . & . & 1 & . & 1 & . & . \\
. & . & . & . & . & . & . & . & 1 & . & . & . & . & 1 & 1 & . & 1 & 1 & . & . & . & 1 & . & 1 & . \end{array} \right).$$ 
and contains only columns of weight one or two. 
%(the matrix representation of $A(11)$ is given in \cite{ArxivPrep})
%. This shows that there is no word of length lower or equal to $11$ mapping three states of the automaton onto a single one.
 On the other hand, the word $abbabbababba$, which is twelve letters long, maps states $q_3$, $q_5$ and $q_9$ on state $q_3$.
\end{proof}
}

Figure \ref{Comparison} represents the SPF of $\mathcal{TR}_9$, compared with the SPF of the automaton of the \v Cern{\'y} family with $9$ states. \rev{For} $t=11$, the SPF of \v Cern{\'y}'s automaton is larger than the SPF of $\mathcal{TR}_9$.

\begin{figure}[!htb]
\begin{center}
% This file was created by matlab2tikz v0.4.7 running on MATLAB 7.14.
% Copyright (c) 2008--2014, Nico Schlömer <nico.schloemer@gmail.com>
% All rights reserved.
% Minimal pgfplots version: 1.3
% 
% The latest updates can be retrieved from
%   http://www.mathworks.com/matlabcentral/fileexchange/22022-matlab2tikz
% where you can also make suggestions and rate matlab2tikz.
% 
\scalebox{1}{
\begin{tikzpicture}

\begin{axis}[%
width=2.52083333333333in,
height=1.565625in,
scale only axis,
xmin=0,
xmax=70,
xlabel={$t$},
ymin=0,
ymax=1.1,
ylabel={$k(t)$},
ylabel style={rotate=-90},
axis x line*=bottom,
axis y line*=left
]

\addplot [color=black,dashed,forget plot]
  table[row sep=crcr]{%
1	0.124999999973383\\
2	0.124999997368377\\
3	0.142857142012872\\
4	0.142857142854098\\
5	0.166666666671915\\
6	0.166666666665606\\
7	0.199999999848444\\
8	0.20000000008551\\
9	0.222222222529055\\
10	0.24999999933658\\
11	0.249999999993378\\
12	0.250000000028081\\
13	0.24999999995822\\
14	0.285714285751482\\
15	0.285714285706632\\
16	0.333333334812778\\
17	0.333333333192627\\
18	0.333333333662907\\
19	0.333333332591508\\
20	0.33333333213406\\
21	0.374999999951399\\
22	0.375000000766292\\
23	0.400000000062562\\
24	0.399999999928497\\
25	0.428571428444798\\
26	0.42857142856144\\
27	0.444444444411872\\
28	0.500000000003922\\
29	0.500000000136168\\
30	0.500000000027029\\
31	0.500000000034447\\
32	0.499999999997669\\
33	0.500000000004093\\
34	0.50000000001404\\
35	0.500000000149811\\
36	0.555555556084045\\
37	0.571428571430772\\
38	0.57142857141281\\
39	0.599999999998033\\
40	0.600000001178699\\
41	0.624999999940485\\
42	0.624999999999957\\
43	0.666666666653896\\
44	0.666666666673194\\
45	0.66666666750163\\
46	0.66666666407383\\
47	0.666666666667766\\
48	0.714285714313291\\
49	0.714285714288053\\
50	0.749999999999716\\
51	0.750000000117836\\
52	0.749999999962029\\
53	0.750000000049113\\
54	0.777777775795627\\
55	0.799999999996174\\
56	0.800000000001091\\
57	0.833333333101336\\
58	0.833333333321661\\
59	0.857142857155623\\
60	0.857142857157612\\
61	0.874999999991843\\
62	0.874999999948727\\
63	0.888888888890264\\
64	1.00000000024282\\
65	0.999999999798092\\
66	0.999999999798092\\
67	0.999999999798092\\
68	0.999999999798092\\
69	0.999999999798092\\
70	0.999999999798092\\
};
\addplot [color=black,solid,forget plot]
  table[row sep=crcr]{%
1	0.125000000019128\\
2	0.142857140133032\\
3	0.166666664999724\\
4	0.199999999999903\\
5	0.199999999999363\\
6	0.199999999997743\\
7	0.199999999995583\\
8	0.199999999556866\\
9	0.222222222225213\\
10	0.222222222222285\\
11	0.222222222379628\\
12	0.249999999995623\\
13	0.285714285609714\\
14	0.333333333483779\\
15	0.333333333359064\\
16	0.375000000000284\\
17	0.400000000002933\\
18	0.400000000022374\\
19	0.40000000000083\\
20	0.428571428571331\\
21	0.444444444444457\\
22	0.444444444445935\\
23	0.444444444448862\\
24	0.499999999981355\\
25	0.555555555555031\\
26	0.555555555838964\\
27	0.555555555553781\\
28	0.600000000014688\\
29	0.599999999976035\\
30	0.636363636336057\\
31	0.666666666646279\\
32	0.666666666707016\\
33	0.666666666567949\\
34	0.700000000000387\\
35	0.750000000002103\\
36	0.777777777777999\\
37	0.777777777782006\\
38	0.800000000000296\\
39	0.80000000002201\\
40	0.833333333303273\\
41	0.857142857139678\\
42	0.874999999887507\\
43	0.888888888889198\\
44	1.00000000024282\\
45	0.999999999798092\\
46	0.999999999798092\\
47	0.999999999798092\\
48	0.999999999798092\\
49	0.999999999798092\\
50	0.999999999798092\\
51	0.999999999798092\\
52	0.999999999798092\\
53	0.999999999798092\\
54	0.999999999798092\\
55	0.999999999798092\\
56	0.999999999798092\\
57	0.999999999798092\\
58	0.999999999798092\\
59	0.999999999798092\\
60	0.999999999798092\\
61	0.999999999798092\\
62	0.999999999798092\\
63	0.999999999798092\\
64	0.999999999798092\\
65	0.999999999798092\\
66	0.999999999798092\\
67	0.999999999798092\\
68	0.999999999798092\\
69	0.999999999798092\\
70	0.999999999798092\\
};
\addplot [color=gray,dotted,forget plot]
  table[row sep=crcr]{%
0	1\\
70	1\\
};
\addplot [color=gray,dotted,forget plot]
  table[row sep=crcr]{%
11	0\\
11	1\\
};
\addplot [color=black,solid,forget plot]
  table[row sep=crcr]{%
0	0.111111111111111\\
1	0.125\\
};
\end{axis}
\end{tikzpicture}}%
\end{center}
\caption{The synchronizing probability function of $\mathcal {TR}_9$ (solid curve), and of the automaton of the \v Cern{\'y} family with $9$ states (dashed curve). We have $k_{\mathcal{TR}_9}(11)=2/9$.}
\label{Comparison}
\end{figure}
We can extend $\mathcal{TR}_9$ to an infinite family of automata with an odd number of
states. Starting from $\mathcal{TR}_n$, let $l_1$ be the letter such that $q_nl_1=q_n$ and $l_2$ the other one. The automaton $\mathcal{TR}_{n+2}$ has $n+2$ states $q_1, \dots, q_{n+2}$, the two letters $l_1$ and $l_2$, and the effects $q_nl_1=q_{n+2}$, $q_{n+2}l_1=q_{n}$, $q_{n-1}l_2=q_{n+1}$, $q_{n+1}l_2=q_{n-1}$, $q_{n+2}l_2=q_{n+2}$ and $q_{n+1}l_1=q_{n+1}$. All the others effects of the letters are the same as in $\mathcal{TR}_n$. Figure \ref{Graphfamily1113} show the automata $\mathcal{TR}_{11}$ and $\mathcal{TR}_{13}$, and Fig. \ref{GeneralFamily} shows $\mathcal{TR}_{2k+7}$, with $k$ odd.

\begin{figure}[!htb]
\begin{center}
\scalebox{0.91}{
\begin{tikzpicture}[->,>=stealth',shorten >=1pt,auto,node distance=1.5cm,
                    semithick]
  \tikzstyle{every state}=[fill=light-gray,draw=none,text=black, scale=0.8]

  \node[state] 			(A)                    {$q_9$};
  \node[state]         	(B) [right of=A] {$q_7$};
  \node[state]         	(C) [ right of=B] {$q_1$};
  \node[state]         	(D) [below left of=C, node distance=2cm] {$q_2$};
  \node[state]         	(E) [below right of=C, node distance=2cm] {$q_3$};
  \node[state]         	(F) [below of=D] {$q_4$};
  \node[state]         	(G) [below of=E] {$q_5$};
  \node[state]         	(H) [below right of=F, node distance=2cm] {$q_6$};
  \node[state]         	(I) [left of=H] {$q_8$};
  \node[state]         	(J) [left of=A] {$q_{11}$};
  \node[state]         	(K) [left of=I] {$q_{10}$};

  \path (A) edge              	node {b} (B)
            edge 				node {} (J)
        (B) edge   				node {} (A)
            edge              	node {a} (C)
        (C) edge              	node {b} (D)
         	edge  				node {} (B)
        (D) edge  			  	node {b} (E)
            edge  			  	node {a} (F)
        (E) edge  			  	node {b} (C)
            edge [loop right]	node {a} (E)
        (F) edge  			  	node {} (D)
            edge 	[swap]			node {b} (G)
        (G) edge  			  	node {a} (E)
            edge 	[swap]			node {b} (H)
        (H) edge  			  	node {a} (I)
            edge 	[swap]			node {b} (F)
        (I) edge  			  	node {} (H)
            edge 				node {b} (K)
        (J) edge  			  	node {a} (A)
            edge [loop left]	node {b} (K) 
        (K) edge  			  	node {} (I)
            edge [loop left]	node {a} (K);   
  \node[state]         	(W) [above right of=E, node distance=2cm] {$q_{13}$};        
  \node[state]         	(U) [right of=W] {$q_{11}$};
  \node[state] 			(L) [right of=U]{$q_9$};                 
  \node[state]         	(M) [right of=L] {$q_7$};
  \node[state]         	(N) [ right of=M] {$q_1$};
  \node[state]         	(O) [below left of=N, node distance=2cm] {$q_2$};
  \node[state]         	(P) [below right of=N, node distance=2cm] {$q_3$};
  \node[state]         	(Q) [below of=O] {$q_4$};
  \node[state]         	(R) [below of=P] {$q_5$};
  \node[state]         	(S) [below right of=Q, node distance=2cm] {$q_6$};
  \node[state]         	(T) [left of=S] {$q_8$}; 
  \node[state]         	(V) [left of=T] {$q_{10}$};
  \node[state]         	(X) [left of=V] {$q_{12}$}; 

  \path (W)edge  			  	node {b} (U)
            edge [loop left]	node {a} (K)
        (U)edge  			  	node {} (W)
            edge 				node {a} (L)  
  		(L) edge              	node {b} (M)
            edge 				node {} (U)
        (M) edge   				node {} (L)
            edge              	node {a} (N)
        (N) edge              	node {b} (O)
         	edge  				node {} (M)
        (O) edge  			  	node {b} (P)
            edge  			  	node {a} (Q)
        (P) edge  			  	node {b} (N)
            edge [loop right]	node {a} (P)
        (Q) edge  			  	node {} (O)
            edge 	[swap]			node {b} (R)
        (R) edge  			  	node {a} (P)
            edge 		[swap]		node {b} (S)
        (S) edge  			  	node {a} (T)
            edge 		[swap]		node {b} (Q)
        (T) edge  			  	node {} (S)
            edge 				node {b} (V)
        (V) edge  			  	node {} (T)
            edge				node {a} (X) 
        (X) edge  			  	node {} (V)
            edge [loop left]	node {b} (V); 
\end{tikzpicture}}
\end{center}
\caption{The automata $\mathcal {TR}_{11}$ and $\mathcal {TR}_{13}$.}
\label{Graphfamily1113}
\end{figure}
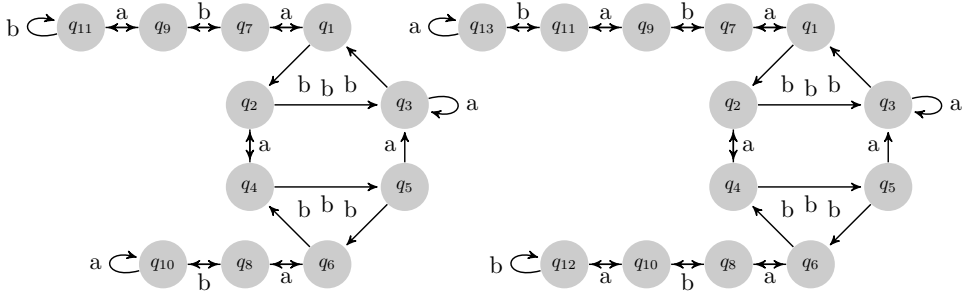

\begin{figure}[!htb]
\begin{center}
\scalebox{1}{
\begin{tikzpicture}[->,>=stealth',shorten >=1pt,auto,node distance=1.5cm,
                    semithick]
  \tikzstyle{every state}=[fill=light-gray,draw=none,text=black, scale=1]

  \node 			(A)                    {$\dots$};
  \node[state]         	(B) [right of=A] {$q_7$};
  \node[state]         	(C) [ right of=B] {$q_1$};
  \node[state]         	(D) [below left of=C, node distance=1.5cm] {$q_2$};
  \node[state]         	(E) [below right of=C, node distance=1.5cm] {$q_3$};
  \node[state]         	(F) [below of=D] {$q_4$};
  \node[state]         	(G) [below of=E] {$q_5$};
  \node[state]         	(H) [below right of=F, node distance=1.5cm] {$q_6$};
  \node[state]         	(I) [left of=H] {$q_8$};
  \node[state]         	(J) [left of= A, node distance=2cm] {$q_{2k+5}$};
  \node         	(K) [left of=I] {$\dots$};
  \node[state]         	(L) [left of =J, node distance=2cm] {$q_{2k+7}$};
  \node[state]         	(M) [left of= K, node distance=2cm] {$q_{2k+6}$};

  \path (A) edge              	node {b} (B)
            edge 			  	node {} (J)
        (B) edge   				node {} (A)
            edge              	node {a} (C)
        (C) edge              	node {b} (D)
         	edge  				node {} (B)
        (D) edge  			  	node {b} (E)
            edge  			  	node {a} (F)
        (E) edge  			  	node {b} (C)
            edge [loop right]	node {a} (E)
        (F) edge  			  	node {} (D)
            edge 			[swap]	node {b} (G)
        (G) edge  			  	node {a} (E)
            edge 			[swap]	node {b} (H)
        (H) edge  			  	node {a} (I)
            edge 			[swap]	node {b} (F)
        (I) edge  			  	node {} (H)
            edge 	node {b} (K)
        (J) edge  			  	node {a} (A)
            edge 	node {} (L)
        (K) edge  			  	node {a} (M)
            edge 	node {} (I)
        (L) edge  	[loop left] 	node {a} (M)
            edge 	node {b} (J)
        (M) edge  	[loop left] 	node {b} (M)
            edge 	node {} (K) ;
\end{tikzpicture}}
\end{center}
\caption{The automaton $\mathcal{TR}_{2k+7}$, with $k$ odd.}
\label{GeneralFamily}
\end{figure}
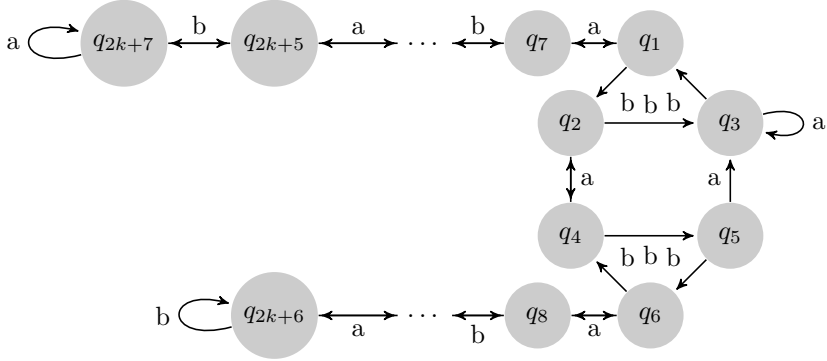

We will now proof the TRT value for $\mathcal{TR}_{2k+7}$. In order to do so, we first provide with a set of lemmas, before proving the result. Let us denote by $(q_i,q_j)$ the vector $(e_i+e_j)^T$. We will analyze the evolution of the matrix $A(t)$ of $\mathcal{TR}_{2k+7}$ with respect to $t$. Notice that in order to obtain $A(t)$ from $A(t-1)$, we multiply the letters by the vectors of $A(t)$ from the right.  We say that a vector $c_{old}$ \emph{can induce} a vector $c_{new}$ if there is a letter $l$ such that $lc_{old}=c_{new}$. We say that a vector $c_{new}$ of $A(t)$ \emph{is induced} by a vector $c_{old}$ of $A(t-1)$ at $t$ if $c_{new}$ is not in $A(t-1)$ and there is a letter $l$ such that $lc_{old}=c_{new}$. As there are two letters in $\mathcal{TR}_n$, each vector can induce at most two vectors. We first notice that:
\begin{lemma}
\label{NewCol}
The reachable vectors induced at $t>1$ are induced by vectors that were themselves induced at $t-1$.
\end{lemma}
\begin{proof}
A vector $c_{new}$ which is in $A(t)$ and not in $A(t-1)$ is equal to $lc_{old}$, for a letter $l$ and some vector $c_{old}$ of $A(t-1)$. If $c_{old}$ was in $A(t-2)$, then $lc_{old}=c_{new}$ would be in $A(t-1)$, therefore $c_{old}$ was induced at $t-1$.
\end{proof}

Therefore, in order to analyze the evolution of $A(t)$ up to $t=T_3$, for each value of $t$, we only have to consider the columns induced at $t-1$. 

The structure of $\mathcal{TR}_n$ provides us with the following property for pairs of states in the tails:
\begin{lemma}
\label{TailEffect}
The vectors that can induce $(q_{2i}, q_{2j-1})$, with $i,j>3$, are also the only vectors that $(q_{2i}, q_{2j-1})$ can induce.
\end{lemma}
\begin{proof}
For any vector $c=(q_{2i}, q_{2j-1})$, with $i,j>3$, we have that $a^2c=b^2c=c$, and if for some column $c'$, we have $ac'=c$ or $bc'=c$, then $a^2c'=b^2c'=c'$. Therefore a vector that can be induced by $c$ can induce $c'$ and vice versa. 
\end{proof}
\begin{cor}
\label{CorOne}
If a vector $(q_{2i}, q_{2j-1})$, with $i,j>3$ is induced at $t$, then it can induce at most one vector at $t+1$.
\end{cor}
\begin{proof}
Any vector can induce at most two vectors. Since a vector $(q_{2i}, q_{2j-1})$, with $i,j>3$ can only be induced by vectors that it can induce, one of the two possible vector that it can induce must already be in $A(t)$.
\end{proof}

We can now prove the result result:

\begin{proposition}
The triple rendezvous time of $\mathcal{TR}_{n}$, with $n=2k+7$ and $k\in \mathbb{N}$, is equal to $n+3$.
\end{proposition}
\begin{proof}
We consider the evolution of $A(t)$ with respect to $t$. 

The $n$ first columns of $A(t)$ are the identity matrix for $t\geq 0$. 

At $t=1$, the column $(q_3,q_5)$ is induced by $e_3$.

At $t=2$, the column $(q_2, q_4)$ is induced by $(q_3,q_5)$. 

At $t=3$, the column $(q_1,q_6)$ is induced by $(q_2,q_4)$. 

At $t=4$, the column $(q_7,q_8)$ is induced by $(q_1,q_6)$.

The column $(q_7,q_8)$ is such that Corollary \ref{CorOne} applies, and therefore induces only one column. For $4\leq t\leq 2k+4$, Corollary \ref{CorOne} applies at every step and the induced column is straightforward. At $t=2k+4$, the column $(q_6,q_9)$ is induced.

At $t=2k+5$, the column $(q_5,q_7)$ is induced by $(q_6,q_9)$.

At $t=2k+6$, the column $(q_4,q_9)$ is induced by $(q_5,q_7)$. 

At $t=2k+7$, the column $(q_2,q_{11})$ and the column $(q_6,q_7)$ are induced by $(q_4,q_9)$\footnote{Starting at $t=2k+7$, the vectors induced are not the same for $\mathcal{TR}_9$, $\mathcal{TR}_{11}$ and $\mathcal{TR}_{13}$, as the vectors including states $q_n$ or $q_{n-1}$ do not induce vectors with higher indices. However it does not change the result.}.

At $t=2k+8$, the column $(q_1,q_{13})$ is induced by $(q_2,q_{11})$, the columns $(q_1,q_8)$ and $(q_5,q_9)$ are induced by $(q_6,q_7)$.

At $t=2k+9$, the column $(q_3,q_{10})$ is induced by $(q_1,q_8)$. The columns $(q_7,q_{15})$ and $(q_3,q_{11})$ are induced by $(q_1,q_{13})$, the column $(q_4,q_7)$ is induced by $(q_5,q_9)$.

At $t=2k+10$, the product of letter $a$ with both the column $(q_3,q_{11})$ or the column $(q_3,q_{10})$ provides with a column of weight three. 

Therefore, the TRT is equal to $2k+10=n+3$, which concludes the proof.
\end{proof}

The graph in Fig. \ref{FamilySync} presents the synchronizing probability function of these automata with $9$, $11$ and $13$ states.
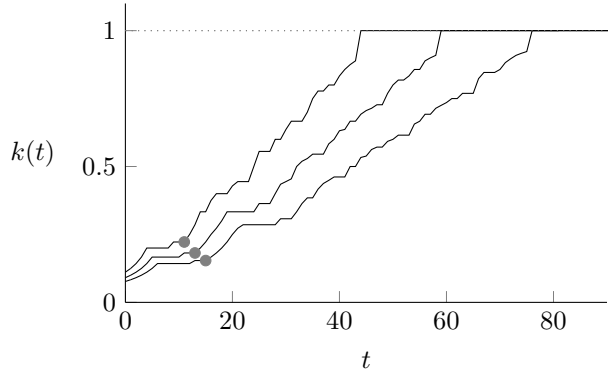
\begin{figure}[!htb]
\begin{center}
% This file was created by matlab2tikz v0.4.7 running on MATLAB 7.14.
% Copyright (c) 2008--2014, Nico Schlömer <nico.schloemer@gmail.com>
% All rights reserved.
% Minimal pgfplots version: 1.3
% 
% The latest updates can be retrieved from
%   http://www.mathworks.com/matlabcentral/fileexchange/22022-matlab2tikz
% where you can also make suggestions and rate matlab2tikz.
% 
\scalebox{1}{
\begin{tikzpicture}
\begin{axis}[%
width=2.52083333333333in,
height=1.565625in,
scale only axis,
xmin=0,
xmax=90,
xlabel={$t$},
ymin=0,
ymax=1.1,
ylabel={$k(t)$},
ylabel style={rotate=-90},
axis x line*=bottom,
axis y line*=left
]
\addplot [color=gray,dotted,forget plot]
  table[row sep=crcr]{%
0	1\\
90	1\\
};
\addplot [color=gray, only marks,mark=*] coordinates { (11,0.222222222128693) };
\addplot [color=gray, only marks,mark=*] coordinates { (13,0.18181818181813) };
\addplot [color=gray, only marks,mark=*] coordinates { (15,0.153846153845677) };
\addplot [color=black,solid,forget plot]
  table[row sep=crcr]{%
1	0.124999999989541\\
2	0.142857142283702\\
3	0.166666666545808\\
4	0.200000003614178\\
5	0.199999999997971\\
6	0.199999999999392\\
7	0.200000000002461\\
8	0.199999999896193\\
9	0.222222222219386\\
10	0.222222222222143\\
11	0.222222222128693\\
12	0.250000000007788\\
13	0.285714285345648\\
14	0.333333333387884\\
15	0.333333333361907\\
16	0.375000000000057\\
17	0.400000000000205\\
18	0.400000000032378\\
19	0.399999999999352\\
20	0.428571428571331\\
21	0.444444444444457\\
22	0.444444444513522\\
23	0.444444444451904\\
24	0.500000000197929\\
25	0.555555555555316\\
26	0.555555555694241\\
27	0.555555555555713\\
28	0.599999999954093\\
29	0.60000000013622\\
30	0.636363636363626\\
31	0.666666666664639\\
32	0.666666666799017\\
33	0.666666666920491\\
34	0.700000000070929\\
35	0.750000000007049\\
36	0.77777777777743\\
37	0.777777777784678\\
38	0.800000000000182\\
39	0.800000000005156\\
40	0.833333333488156\\
41	0.857142857146187\\
42	0.87500000003827\\
43	0.888888885750958\\
44	1.00000000024282\\
45	0.999999999798092\\
46	0.999999999798092\\
47	0.999999999798092\\
48	0.999999999798092\\
49	0.999999999798092\\
50	0.999999999798092\\
51	0.999999999798092\\
52	0.999999999798092\\
53	0.999999999798092\\
54	0.999999999798092\\
55	0.999999999798092\\
56	0.999999999798092\\
57	0.999999999798092\\
58	0.999999999798092\\
59	0.999999999798092\\
60	0.999999999798092\\
61	0.999999999798092\\
62	0.999999999798092\\
63	0.999999999798092\\
64	0.999999999798092\\
65	0.999999999798092\\
66	0.999999999798092\\
67	0.999999999798092\\
68	0.999999999798092\\
69	0.999999999798092\\
70	0.999999999798092\\
71	0.999999999798092\\
72	0.999999999798092\\
73	0.999999999798092\\
74	0.999999999798092\\
75	0.999999999798092\\
76	0.999999999798092\\
77	0.999999999798092\\
78	0.999999999798092\\
79	0.999999999798092\\
80	0.999999999798092\\
81	0.999999999798092\\
};
\addplot [color=black,solid,forget plot]
  table[row sep=crcr]{%
1	0.100000000027379\\
2	0.111111111111128\\
3	0.125000000218932\\
4	0.142857143340507\\
5	0.166666666666487\\
6	0.1666666666666\\
7	0.166666666669101\\
8	0.166666666668391\\
9	0.16666666627745\\
10	0.1666666666666\\
11	0.181818181310518\\
12	0.181818181818073\\
13	0.18181818181813\\
14	0.200000000032105\\
15	0.222222222210576\\
16	0.24999999999639\\
17	0.27272727277284\\
18	0.299999999969202\\
19	0.333333333333371\\
20	0.333333333323708\\
21	0.333333333332718\\
22	0.333333333333371\\
23	0.333333333333456\\
24	0.33333333338453\\
25	0.363636363930539\\
26	0.363636363622987\\
27	0.363636363637397\\
28	0.399999999999693\\
29	0.43478260869631\\
30	0.444444444443604\\
31	0.454545454546519\\
32	0.500000000281773\\
33	0.519230769231399\\
34	0.529411764800443\\
35	0.545454546937549\\
36	0.545454545454049\\
37	0.545454545454902\\
38	0.583333333332575\\
39	0.600000000050045\\
40	0.63157894744495\\
41	0.636363636366241\\
42	0.666666666666174\\
43	0.666666666667936\\
44	0.692307692307622\\
45	0.705882352940989\\
46	0.714285714296921\\
47	0.727272727272691\\
48	0.727272727271782\\
49	0.769230769500268\\
50	0.800000000061345\\
51	0.818181818182438\\
52	0.818181818181216\\
53	0.833333333332462\\
54	0.857142857887595\\
55	0.857142857175688\\
56	0.888888888885717\\
57	0.899999999951774\\
58	0.90909090908859\\
59	1.00000000224479\\
60	1.00000000002814\\
61	1.00000000002814\\
62	1.00000000002814\\
63	1.00000000002814\\
64	1.00000000002814\\
65	1.00000000002814\\
66	1.00000000002814\\
67	1.00000000002814\\
68	1.00000000002814\\
69	1.00000000002814\\
70	1.00000000002814\\
71	1.00000000002814\\
72	1.00000000002814\\
73	1.00000000002814\\
74	1.00000000002814\\
75	1.00000000002814\\
76	1.00000000002814\\
77	1.00000000002814\\
78	1.00000000002814\\
79	1.00000000002814\\
80	1.00000000002814\\
81	1.00000000002814\\
82	1.00000000002814\\
83	1.00000000002814\\
84	1.00000000002814\\
85	1.00000000002814\\
86	1.00000000002814\\
87	1.00000000002814\\
88	1.00000000002814\\
89	1.00000000002814\\
90	1.00000000002814\\
91	1.00000000002814\\
92	1.00000000002814\\
93	1.00000000002814\\
94	1.00000000002814\\
95	1.00000000002814\\
96	1.00000000002814\\
97	1.00000000002814\\
98	1.00000000002814\\
99	1.00000000002814\\
100	1.00000000002814\\
101	1.00000000002814\\
102	1.00000000002814\\
103	1.00000000002814\\
104	1.00000000002814\\
105	1.00000000002814\\
106	1.00000000002814\\
107	1.00000000002814\\
108	1.00000000002814\\
109	1.00000000002814\\
110	1.00000000002814\\
111	1.00000000002814\\
112	1.00000000002814\\
113	1.00000000002814\\
114	1.00000000002814\\
115	1.00000000002814\\
116	1.00000000002814\\
117	1.00000000002814\\
118	1.00000000002814\\
119	1.00000000002814\\
120	1.00000000002814\\
121	1.00000000002814\\
};
\addplot [color=black,solid,forget plot]
  table[row sep=crcr]{%
1	0.0833333333113728\\
2	0.0909090906283865\\
3	0.100000000000023\\
4	0.111111111971468\\
5	0.124999997403052\\
6	0.142857142857025\\
7	0.142857142857224\\
8	0.142857142877347\\
9	0.142857142878029\\
10	0.14285714112043\\
11	0.142857142811749\\
12	0.14285714272981\\
13	0.153846153847809\\
14	0.153846154360309\\
15	0.153846153845677\\
16	0.166666666666856\\
17	0.181818181818244\\
18	0.200000000003229\\
19	0.222222222222825\\
20	0.249999999999204\\
21	0.272727272727366\\
22	0.285714285909592\\
23	0.285714285714221\\
24	0.285714285703705\\
25	0.285714285709957\\
26	0.285714285714562\\
27	0.285714285717347\\
28	0.285714286066423\\
29	0.307692307692406\\
30	0.307692307692435\\
31	0.307692307692776\\
32	0.333333333335759\\
33	0.363636363643764\\
34	0.384615385234838\\
35	0.384615384615472\\
36	0.416666666669698\\
37	0.437500000402565\\
38	0.448979591840839\\
39	0.461538461613713\\
40	0.461538461539362\\
41	0.461538461542318\\
42	0.499999999999204\\
43	0.500000000010175\\
44	0.533333333335236\\
45	0.538461538462172\\
46	0.562500000251362\\
47	0.571428571433103\\
48	0.571428571428555\\
49	0.592592592762344\\
50	0.600000000004911\\
51	0.615384615383846\\
52	0.615384615394305\\
53	0.615384615399648\\
54	0.656250000000341\\
55	0.666666666666629\\
56	0.692307694636952\\
57	0.692307692320355\\
58	0.714285714288053\\
59	0.724999999998204\\
60	0.735294117653666\\
61	0.749999999791498\\
62	0.750000000005002\\
63	0.769230769247883\\
64	0.769230769265278\\
65	0.769230769221622\\
66	0.823529412114226\\
67	0.846153846150401\\
68	0.846153846177685\\
69	0.846153846156739\\
70	0.857142857456381\\
71	0.88235294117851\\
72	0.896551724138931\\
73	0.909090909022694\\
74	0.916666667020323\\
75	0.923076923076067\\
76	1.00000000224371\\
77	1.00000000003652\\
78	1.00000000003652\\
79	1.00000000003652\\
80	1.00000000003652\\
81	1.00000000003652\\
82	1.00000000003652\\
83	1.00000000003652\\
84	1.00000000003652\\
85	1.00000000003652\\
86	1.00000000003652\\
87	1.00000000003652\\
88	1.00000000003652\\
89	1.00000000003652\\
90	1.00000000003652\\
91	1.00000000003652\\
92	1.00000000003652\\
93	1.00000000003652\\
94	1.00000000003652\\
95	1.00000000003652\\
96	1.00000000003652\\
97	1.00000000003652\\
98	1.00000000003652\\
99	1.00000000003652\\
100	1.00000000003652\\
101	1.00000000003652\\
102	1.00000000003652\\
103	1.00000000003652\\
104	1.00000000003652\\
105	1.00000000003652\\
106	1.00000000003652\\
107	1.00000000003652\\
108	1.00000000003652\\
109	1.00000000003652\\
110	1.00000000003652\\
111	1.00000000003652\\
112	1.00000000003652\\
113	1.00000000003652\\
114	1.00000000003652\\
115	1.00000000003652\\
116	1.00000000003652\\
117	1.00000000003652\\
118	1.00000000003652\\
119	1.00000000003652\\
120	1.00000000003652\\
121	1.00000000003652\\
122	1.00000000003652\\
123	1.00000000003652\\
124	1.00000000003652\\
125	1.00000000003652\\
126	1.00000000003652\\
127	1.00000000003652\\
128	1.00000000003652\\
129	1.00000000003652\\
130	1.00000000003652\\
131	1.00000000003652\\
132	1.00000000003652\\
133	1.00000000003652\\
134	1.00000000003652\\
135	1.00000000003652\\
136	1.00000000003652\\
137	1.00000000003652\\
138	1.00000000003652\\
139	1.00000000003652\\
140	1.00000000003652\\
141	1.00000000003652\\
142	1.00000000003652\\
143	1.00000000003652\\
144	1.00000000003652\\
145	1.00000000003652\\
146	1.00000000003652\\
147	1.00000000003652\\
148	1.00000000003652\\
149	1.00000000003652\\
150	1.00000000003652\\
151	1.00000000003652\\
152	1.00000000003652\\
153	1.00000000003652\\
154	1.00000000003652\\
155	1.00000000003652\\
156	1.00000000003652\\
157	1.00000000003652\\
158	1.00000000003652\\
159	1.00000000003652\\
160	1.00000000003652\\
161	1.00000000003652\\
162	1.00000000003652\\
163	1.00000000003652\\
164	1.00000000003652\\
165	1.00000000003652\\
166	1.00000000003652\\
167	1.00000000003652\\
168	1.00000000003652\\
169	1.00000000003652\\
};
\addplot [color=gray,dotted,forget plot]
  table[row sep=crcr]{%
0	1\\
90	1\\
};
\addplot [color=black,solid,forget plot]
  table[row sep=crcr]{%
0	0.111111111111111\\
1	0.125\\
};
\addplot [color=black,solid,forget plot]
  table[row sep=crcr]{%
0	0.0909090909090909\\
1	0.1\\
};
\addplot [color=black,solid,forget plot]
  table[row sep=crcr]{%
0	0.0769230769230769\\
1	0.0833333333333333\\
};
\end{axis}
\end{tikzpicture}
}%
\end{center}
\caption{\rev{The synchronizing probability function of $\mathcal {TR}_9$, $\mathcal {TR}_{11}$ and $\mathcal {TR}_{13}$ (from the left to the right), and the points at which Conjecture \ref{ConjT3} is not satisfied.}}
\label{FamilySync}
\end{figure}

\section{Conclusion}
In this paper, we first formalised the concept of triple rendezvous time as an intermediate step towards \v Cern{\'y}'s conjecture. Using the synchronizing probability function, a tool which allows to represent the synchronization process within an automaton, we were able to prove a non-trivial upper bound on the triple rendezvous time.

Then, by providing an infinite family of automata for which $T_3=n+3$ (with $n$ being the number of states of the automaton), we refuted Conjecture \ref{ConjK}, formulated in \cite{jungers_sync_12}. 
Conjecture \ref{ConjK} was stated as a tentative roadmap towards a proof of \v Cern{\'y}'s conjecture with the help of the synchronizing probability function, and in that sense our conterexample is a negative result towards that direction.

A natural continuation of this research would be to find non-trivial bounds for $T_l$, with $3<l\leq n$ (i.e., the smallest number such that the set of reachable vectors includes a column of weight at least $l$). Another research question is how to narrow the gap between $n+3$ and $n^2/(6.4...)$ for the triple rendezvous time.

\section*{Acknowledgements}
We would like to thank Matthew Philippe, Myriam Gonze, Bal\' asz Gerencs\' er, \' Elodie Boucquey and Jean Boucquey for their helpful comments and interesting discussions, Vladimir Gusev for proofreading, and the anonymous reviewers for their suggestions.

\bibliography{references}

\def\cprime{$'$} \newcommand{\noopsort}[1]{} \newcommand{\singleletter}[1]{#1}
  \def\cprime{$'$}
\begin{thebibliography}{10}

\bibitem{AnanGusVolk}
{\sc D.~S. Ananichev, V.~V. Gusev, and M.~V. Volkov}, {\em Slowly synchronizing
  automata and digraphs}, Mathematical Foundations of Computer Science 2010,
  (2010), pp.~55--65.

\bibitem{1065155}
{\sc D.~S. Ananichev and M.~V. Volkov}, {\em Synchronizing generalized
  monotonic automata}, Theoretical Computer Science, 330 (2005), pp.~3--13.

\bibitem{Ananichev200730}
{\sc D.~S. Ananichev, M.~V. Volkov, and Yu.~I. Zaks}, {\em Synchronizing
  automata with a letter of deficiency 2}, Theoretical Computer Science, 376
  (2007), pp.~30--41.

\bibitem{BBP11}
{\sc M.-P. B\'{e}al, M.~V. Berlinkov, and D.~Perrin}, {\em A quadratic upper
  bound on the size of a synchronizing word in one-cluster automata},
  International Journal of Foundations of Computer Science, 22 (2011),
  pp.~277--288.

\bibitem{Berlinkov-extension}
{\sc M.~V. Berlinkov}, {\em On a conjecture by {C}arpi and {D}'{A}lessandro},
  International Journal of Foundations of Computer Science, 22 (2011),
  pp.~1565--1576.

\bibitem{JungersBlondelOlshevsky14}
{\sc V.~Blondel, R.~M. Jungers, and A.~Olshevsky}, {\em On primitivity of sets
  of matrices}, Automatica, volume 61 (2015), pp.~80--88.

\bibitem{BoydBook}
{\sc S.~Boyd and L.~Vandenberghe}, {\em Convex optimization}, Cambridge
  University Press, 2004.

\bibitem{ThesisSp}
{\sc A.~F.~P. Cardoso}, {\em The \v Cern{\'y} Conjecture and Other
  Synchronization Problems}, phd. thesis, Departamento de Matematica Faculdade
  de Ciencias da Universidade do Porto, 2014.

\bibitem{carpiarticle}
{\sc A.~Carpi and F.~D'Alessandro}, {\em Independent sets of words and the
  synchronization problem}, Advances in Applied Mathematics, 50, issue 3 (March
  2013), pp.~339--355.

\bibitem{cerny64}
{\sc J.~{\v C}ern{\'y}}, {\em Pozn{\'a}mka k homog\'ennym eksperimentom s
  kone{\v c}n{\'y}mi automatami}, Matematicko-fysikalny Casopis SAV, 14 (1964),
  pp.~208--216.

\bibitem{cernyPirickaRosenauerova64}
{\sc J.~{\v C}ern{\'y}, A.~Pirick\'a, and B.~Rosenauerova}, {\em On directable
  automata}, Kybernetica, 7 (1971), pp.~289--298.

\bibitem{PYChev}
{\sc P.-Y. Chevalier, J.~M. Hendrickx, and R.~M. Jungers}, {\em Reachability of
  consensus and synchronizing automata}, ArXiv preprint.
  http://arxiv.org/pdf/1505.00144v1.pdf,  (2015).

\bibitem{Dubuc98}
{\sc L.~Dubuc}, {\em Sur les automates circulaires et la conjecture de
  \v{C}ern\'{y}}, RAIRO Informatique Theorique et Appliqu\'ee, 32 (1998),
  pp.~21--34.

\bibitem{eppstein90reset}
{\sc D.~Eppstein}, {\em Reset sequences for monotonic automata}, SIAM Journal
  on Computing, 19 (1990), pp.~500--510.

\bibitem{Frankl82}
{\sc P.~Frankl}, {\em An extremal problem for two families of sets}, European
  Journal of Combinatorics, 3 (1982), pp.~125--127.

\bibitem{GonzeLATA}
{\sc F.~Gonze and R.~M. Jungers}, {\em On the synchronizing probability
  function and the triple rendezvous time: New approaches to \v{C}ern{\'y}'s
  conjecture}, in LATA 2015, vol.~8977 of Lecture Notes in Computer Science,
  Springer-Verlag, 2015, pp.~212--223.

\bibitem{GJT2014}
{\sc F.~Gonze, R.~M. Jungers, and A.~N. Trahtman}, {\em A note on a recent
  attempt to improve the {P}in-{F}rankl bound}, Discrete Mathematics and
  Theoretical Computer Science, 17 (2015), pp.~307--308.

\bibitem{jungers_sync_12}
{\sc R.~M. Jungers}, {\em The synchronizing probability function of an
  automaton}, SIAM Journal on Discrete Mathematics, 26 (2012), pp.~177--192.

\bibitem{JKari01}
{\sc J.~Kari}, {\em A counter example to a conjecture concerning synchronizing
  words in finite automata}, EATCS Bulletin, 73 (2001), p.~146.

\bibitem{kari03}
{\sc J.~Kari}, {\em Synchronizing finite automata on eulerian digraphs},
  Theoretical Computer Science, 295 (2003), pp.~223--232.

\bibitem{KlyachkoRystsovSpivak87}
{\sc A.~A. Klyachko, I.~K. Rystsov, and M.~A. Spivak}, {\em An extremal
  combinatorial problem associated with the bound on the length of a
  synchronizing word in an automaton}, Kibernetika, 2 (1987), pp.~16--20.

\bibitem{OlUm}
{\sc J.~Olschewski and M.~Ummels}, {\em The complexity of finding reset words
  in finite automata}, Mathematical Foundations of Computer Science 2010,
  (2010), pp.~568--579.

\bibitem{Pin83a}
{\sc J.-E. Pin}, {\em On two combinatorial problems arising from automata
  theory}, Annals of Discrete Mathematics, 17 (1983), pp.~535--548.

\bibitem{Roman}
{\sc A.~Roman}, {\em A note on \v{C}ern{\'y} conjecture for automata over
  3-letter alphabet}, Journal of Automata, Languages and Combinatorics, Volume
  13 Issue 2 (2008), pp.~141--143.

\bibitem{skvortsov2011experimental}
{\sc E.~Skvortsov and E.~Tipikin}, {\em Experimental study of the shortest
  reset word of random automata}, in CIAA 2011, vol.~6807 of Lecture Notes in
  Computer Science, Springer-Verlag, 2011, pp.~290--298.

\bibitem{steinberg-2009}
{\sc B.~Steinberg}, {\em The averaging trick and the {{\v C}}ern{\'y}
  conjecture}, International Journal of Foundations of Computer Science, 22
  (2011), pp.~1697--1706.

\bibitem{trahtman_cerny}
{\sc A.~N. Trahtman}, {\em The {{\v C}}ern{\'y} conjecture for aperiodic
  automata}, Discrete mathematics and Theoretical Computer Science, 9 (2007),
  pp.~3--10.

\bibitem{trahtman}
\leavevmode\vrule height 2pt depth -1.6pt width 23pt, {\em The road coloring
  problem}, Israel Journal of Mathematics, 172 (2009), pp.~51--60.

\bibitem{Trahtman2011}
\leavevmode\vrule height 2pt depth -1.6pt width 23pt, {\em Modifying the upper
  bound on the length of minimal synchronizing word}, in FCT 2011, vol.~6914 of
  Lecture Notes in Computer Science, Springer-Verlag, 2011, pp.~173--180.

\bibitem{volkov_survey}
{\sc M.~V. Volkov}, {\em Synchronizing automata and the \v {C}ern{\'y}
  conjecture}, in LATA 2008, vol.~5196 of Lecture Notes in Computer Science,
  Springer-Verlag, 2008, pp.~11--27.

\bibitem{Vorel}
{\sc V.~Vorel and A.~Roman}, {\em Parameterized complexity of synchronization
  and road coloring}, Discrete Mathematics and Theoretical Computer Science, 17
  (2015), pp.~283--306.

\end{thebibliography}

\end{document}